\documentclass[aps,10pt,twocolumn,showpacs,pra,citeautoscript,amsmath,amssymb,floatfix,superscriptaddress]{revtex4-1}

\usepackage{amsmath}
\usepackage{amssymb}
\usepackage{epsfig}
\usepackage{color}
\usepackage{amsthm}
\usepackage{hyperref}
\usepackage{appendix}

\newtheorem{theorem}{Theorem}
\newtheorem*{theorem*}{Theorem}
\newtheorem{proposition}[theorem]{Proposition}
\newtheorem*{proposition*}{Proposition}
\newtheorem{lemma}[theorem]{Lemma}
\newtheorem*{lemma*}{Lemma}
\newtheorem{corollary}[theorem]{Corollary}
\newtheorem{definition}[theorem]{Definition}
\newtheorem*{duplicate*}{}

\newcommand{\ket}[1]{|#1\rangle}

\newcommand{\ketbra}[2]{|#1\rangle\langle #2|}
\newcommand{\tr}[0]{\textnormal{Tr}}

\begin{document}
\title{Clean quantum and classical communication protocols}
\author{Harry Buhrman}
\affiliation{QuSoft, CWI Amsterdam and University of Amsterdam, Science Park 123, 1098 XG Amsterdam, Netherlands}
\author{Matthias Christandl}
\affiliation{QMATH, Department of Mathematical Sciences, University of Copenhagen, Universitetsparken 5, 2100 Copenhagen, Denmark}
\author{Christopher Perry}
\affiliation{QMATH, Department of Mathematical Sciences, University of Copenhagen, Universitetsparken 5, 2100 Copenhagen, Denmark}
\author{Jeroen Zuiddam}
\affiliation{QuSoft, CWI Amsterdam and University of Amsterdam, Science Park 123, 1098 XG Amsterdam, Netherlands}

\begin{abstract}
By how much must the communication complexity of a function increase if we demand that the parties not only correctly compute the function but also return all registers (other than the one containing the answer) to their initial states at the end of the communication protocol? Protocols that achieve this are referred to as \emph{clean} and the associated cost as the \emph{clean communication complexity}. Here we present clean protocols for calculating the Inner Product of two $n$-bit strings, showing that (in the absence of pre-shared entanglement) at most $n+3$ qubits or $n+O(\sqrt{n})$ bits of communication are required. The quantum protocol provides inspiration for obtaining the optimal method to implement distributed \emph{CNOT} gates in parallel whilst minimizing the amount of quantum communication. For more general functions, we show that nearly all Boolean functions require close to $2n$ bits of classical communication to compute and close to $n$ qubits if the parties have access to pre-shared entanglement. Both of these values are maximal for their respective paradigms.
\end{abstract}
\maketitle

\emph{Introduction.} In a communication task two players, traditionally named Alice and Bob, receive inputs $x$ and $y$ and wish to calculate the value of some function $f$. To achieve this, messages will have to be exchanged between them and, depending on the resources available to them, these may consist of classical or quantum communication in the form of bits and qubits respectively. Typically in such scenarios one is interested in minimizing the amount of communication that has to take place to evaluate the function and the number of bits/qubits that must be exchanged to do this is referred to as the classical/quantum \emph{communication complexity} \cite{yao1979some,yao1993quantum}.

A protocol for calculating a function will act on three distinct types of registers. Each player will receive an input register, containing $x$ or $y$, and an ancillary working space, initialized in some standard state such as a string of bits all set to $0$, a number of qubits provided in the $\ket{0}$ state or possibly containing entangled states shared between the parties. The final type of register is the answer register which will contain the value of $f\left(x,y\right)$ at the end of the protocol. On the completion of a generic protocol for computing $f$, the input and ancillary registers will no longer be in their starting states and will depend upon both $x$ and $y$.

However, leaving these registers in such states can be problematic. Firstly, if Alice and Bob wish to keep private the particular protocol that they ran, then discarding these unclean states may leak information regarding this to a third party. Secondly, in the quantum setting, if the players wish to run the protocol over a superposition of input states (perhaps as a subroutine of a larger computation), then allowing the ancillary registers to end up in some unclean, input dependent state and then discarding them can lead to a loss of coherence in the superposition over answers. Finally, the players' computational space may be in short supply and without knowing the registers' final states they cannot easily use them for future calculations.

To avoid such issues we can demand that a protocol (in addition to computing $f$) returns the input and ancillary registers to their starting state. Following \cite{cleve1999quantum}, we call such a protocol \emph{clean} and the minimum number of bits/qubits that a clean protocol needs to exchange to compute a given function is the \emph{clean communication complexity}. We shall denote these quantities by $\textit{C}_{\textit{clean}}\left(f\right)$ and $\textit{Q}_{\textit{clean}}\left(f\right)$. In the case where the players have access to pre-shared entanglement (which they must restore at the end of the protocol), the associated cost will be written $Q^{*}_{\textit{clean}}$. We focus on the scenario where the players must compute the function exactly.

In all three scenarios, an unclean communication protocol can be converted into a clean one at the cost of doubling the communication. To do this, the players run the unclean protocol, copy the output to another location and then run the unclean protocol backwards. At first glance it may appear that clean, classical protocols are even easier to construct: the players keep a copy of their input and then simply erase all ancillary bits once the protocol is complete. However, Landauer's principle \cite{landauer1961irreversibility, bennett1973logical, bennett2003notes} implies that such irreversible manipulations will generate heat or else cost work. As such, if one is interested in avoiding such costs, it makes sense to consider protocols where all operations must be reversible. In light of these constructions, it is natural to ask: do more efficient clean protocols, without this doubling in communication, exist?

In the first part of this paper we focus on the clean communication complexity of computing the Inner Product of two distributed bit strings of length $n$, showing that (in the absence of pre-shared entanglement) this can be done by exchanging $n+3$ qubits. As a clean protocol for this function must exchange at least $n+1$ qubits, this is very close to tight. We also provide a clean, classical protocol that computes Inner Product while exchanging only $n+\textit{O}\left(\sqrt{n}\right)$ bits. This provides a saving over the most obvious protocol which, as we shall show, are close to optimal for the clean, classical computation of most functions.

A variation on our quantum protocol can be used to implement $n$ copies of a \emph{CNOT} gate in parallel by exchanging $n+1$ qubits. In a quantum computing architecture consisting of distributed clusters of highly controllable qubits linked by quantum communication (such as that envisaged in \cite{kielpinski2002architecture}), it is prudent to minimize the number of qubits exchanged. Our implementation is optimal.

Next we turn to the clean communication complexity of random functions on inputs of length $n$. We show here that in contrast to Inner Product, nearly all functions are such that $\textit{C}_{\textit{clean}}\left(f\right)$ is close to the maximal $2n$: the simple method of generating clean protocols discussed above is near optimal. On the quantum side, we find that $Q^{*}_{\textit{clean}}\left(f\right)$ is close to $n$ for most functions. As superdense coding \cite{bennett1992communication}  allows all functions to be uncleanly computed while exchanging $\frac{n}{2}$ qubits when the players pre-share entanglement, this is again close to maximal. Whether similarly $Q_{\textit{clean}}\left(f\right)$ is close to $2n$ remains an open question.

\emph{Clean Protocols.} Clean protocols have a long history in proving bounds in the model of quantum communication complexity with free entanglement assistance \cite{cleve1997substituting}. For example, considering clean, quantum protocols for the Inner Product function was used to imply that any entanglement assisted quantum protocol for this function must use at least $\left\lceil n/2\right\rceil$ qubits \cite{cleve1999quantum}. By making use of superdense coding to transmit one player's input to the other, this bound can be achieved. Clean protocols have also been used to show a lower bound on the entanglement assisted, quantum communication complexity \cite{buhrman2001communication} and that, in this model of communication, most functions have complexity that scales linearly in $n$ \cite{montanaro2007lower}. Cleanliness has also been used to analyze privacy amongst honest players \cite{klauck2002quantum}, bound the amount of quantum communication required to implement distributed quantum computation \cite{nielsen2003quantum} and for constructing resource inequalities that carefully account for the way protocols can be combined \cite{harrow2010time,harrow2009entanglement}.

More formally, a clean, quantum protocol for computing a function $f:\left\{0,1\right\}^n\times\left\{0,1\right\}^n\rightarrow\left\{0,1\right\}$ is defined as follows \cite{cleve1999quantum}. The initial state at the beginning of the protocol is of the form:
\begin{equation}
\ket{x}_A\ket{\vec{0}}_{A_0}\ket{y}_B\ket{\vec{0}}_{B_0}\ket{\Phi}_{A_EB_E}\ket{z}_{B_{\textrm{ans}}},
\end{equation}
where $\ket{x}_A=\bigotimes_{i=1}^{n}\ket{x_i}_{A_i}$ and $\ket{y}_B=\bigotimes_{i=1}^{n}\ket{y_i}_{B_i}$ are Alice and Bob's respective inputs stored in $n$ qubits, $\ket{\vec{0}}_{A_0}$ and $\ket{\vec{0}}_{B_0}$ their qubit ancillas, $\ket{\Phi}_{A_EB_E}$ their pre-shared entanglement (if supplied) and $\ket{z}_{B_{\textrm{ans}}}$ is the initial state of the answer register with $z\in\left\{0,1\right\}$. Throughout this paper we will assume that at the beginning and end of a protocol the answer register is held by Bob.

Players then take turns to act on their share of the qubits. In each turn a player will apply a unitary transformation to the qubits in their possession and then send some subset of them to the other player. The protocol computes $f$ cleanly if the final state of the qubits is:
\begin{equation}
\ket{x}_A\ket{\vec{0}}_{A_0}\ket{y}_B\ket{\vec{0}}_{B_0}\ket{\Phi}_{A_EB_E}\ket{z\oplus f\left(x,y\right)}_{B_{\textrm{ans}}},
\end{equation}
where the addition in the answer register is modulo 2. Clean classical protocols are defined similarly but with registers and communication given in terms of bits rather than qubits and no entanglement. All transformations must be reversible.

\emph{Inner Product.} The specific function that we shall focus on in this paper is the Inner Product function, $\textit{IP}_n$. This is defined by:
\begin{align}
\begin{split}
&\textit{IP}_n:\left\{0,1\right\}^n\times\left\{0,1\right\}^n\rightarrow\left\{0,1\right\},\\
&\textit{IP}_n\left(x,y\right)=\sum_{i=1}^{n} x_i\cdot y_i \quad \textrm{mod 2}.
\end{split}
\end{align}
It is well known that for both players to know the answer, at least $n$ bits of classical communication are needed to (uncleanly) compute $\textit{IP}_n$ exactly \cite[Example 1.29]{kushilevitz1997communication}. For quantum strategies in which the players pre-share entanglement, $\left\lceil\frac{n}{2}\right\rceil$ qubits must be sent \cite{cleve1999quantum} to achieve the same goal. In \cite{cleve1999quantum}, it is also shown that clean, quantum protocols for computing $\textit{IP}_n$ must exchange at least $n$ qubits. The quantum communication required to uncleanly compute $\textit{IP}_n$ without the help of prior entanglement is unknown (though must lie between $\left\lceil\frac{n}{2}\right\rceil$ and $n$). For quantum protocols that are allowed to err with fixed probability less than $1/2$, the complexity is still $\Omega\left(n\right)$ \cite{kremer1995quantum}.

Here we shall examine the clean communication complexity of $\textit{IP}_n$ without entanglement assistance. To this end, we first consider the quantum communication complexity of implementing the transformation:
\begin{equation} \label{eq:IP in phase}
\ket{x}_A\ket{y}_B\mapsto \left(-1\right)^{x\cdot y}\ket{x}_A\ket{y}_B,
\end{equation}
i.e. the distributed computation of the inner product of $x$ and $y$ in the phase. Such a transformation corresponds to performing \emph{controlled-Z} gates across $n$ pairs of qubits and by a suitable local basis change this can be converted into an implementation of $n$-fold \emph{CNOTs}.

In \cite{harrow2004coherent} it was shown that $2$ qubits of communication together with sharing $4$ ebits is exactly equivalent as a resource to the ability to implement $2$ \emph{CNOT} gates and sharing $4$ ebits. As such, this provides a protocol for implementing $\textit{IP}_n$ in the phase using $n+8$ qubits of communication and 8 ancilla qubits (for even $n$). This can be adapted to give a protocol requiring $n+2$ qubits of communication for even $n$ and $n+3$ qubits when $n$ is odd. In the following lemma, we give an improved, optimal protocol:
\begin{lemma} \label{le:IP in phase}
The clean, quantum communication complexity of exactly implementing $\textit{IP}_n$ in the phase satisfies:
\begin{equation}
\textit{Q}_{\textit{clean}}\left(\textit{IP}_{n}^{\textit{phase}}\right)=n+1.
\end{equation}
One ancilla qubit is required.
\end{lemma}
(Without using ancilla qubits, $n+1$ qubits for odd $n$ and $n+2$ for even $n$ suffice.)
\begin{proof}
The $n+1$ qubit protocol for even $n$ is as follows. Alice initially prepares an ancilla qubit in the state $\ket{x_1}$ and sends it to Bob who applies a phase of $\left(-1\right)^{x_1\cdot y_1}$. He then adds $y_2$ to the communication qubit and sends it back to Alice in the state $\ket{x_1+y_2}$. Now, Alice cleans up her previous communication by subtracting $x_1$ from the communication and then uses the value of $y_2$ to apply the phase $\left(-1\right)^{x_2\cdot y_2}$. She then adds $x_3$ to the communication qubit to leave it in the state $\ket{y_2\oplus x_3}$ and sends it back to Bob. A schematic of these first rounds is given in Figure \ref{fig:PhaseIP}.
\begin{figure}
\centering
\includegraphics[width=.9\linewidth]{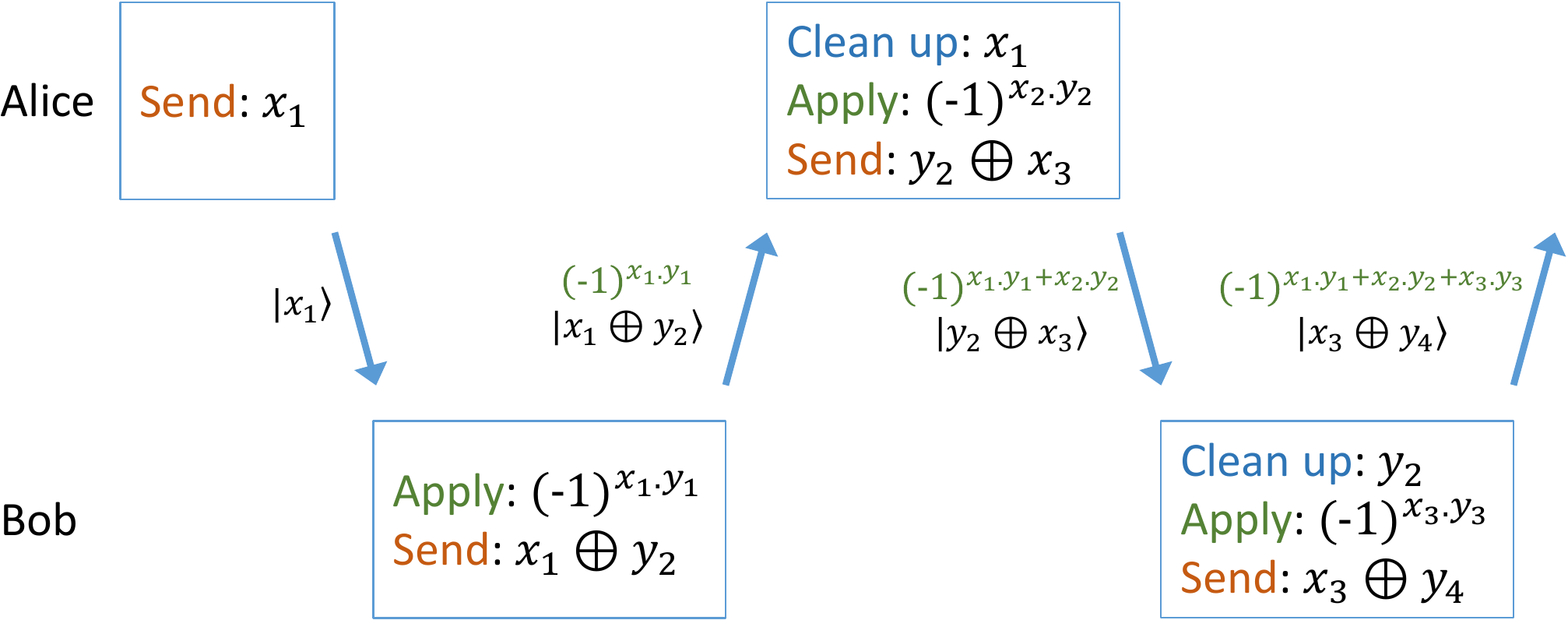}
\caption{\emph{Clean, quantum protocol for calculating $\textit{IP}_n$ in the phase.} Here we illustrate the first 4 rounds of communication. In each round, a player cleans up the message they sent previously, applies the relevant global phase and communicates the next bit of their input string.}
\label{fig:PhaseIP}
\end{figure}

The players then proceed similarly, with each round of communication being used to convey a new bit to the other party and send a received bit back in order to clean the ancilla qubit. After $n$ rounds, the global phase will be $\left(-1\right)^{x\cdot y}$ and Alice will hold the communication qubit in the state $\ket{y_n}$. She sends this back to Bob who cleans it, completing the protocol using $n+1$ qubits of communication and the change in ownership of one ancilla qubit. For odd $n$, Alice will perform the final cleaning step. The full protocol to implement the transformation without an ancilla qubit is given in Appendix \ref{ap:IP in phase}. 

The proof of the lower bound is in Appendix~\ref{ap:phase lb}. It is based upon the concept of information complexity \cite{touchette2015quantum} and showing that in a clean protocol for implementing Eq.~\eqref{eq:IP in phase} $n$ bits of information must flow in each direction. In the absence of pre-shared entanglement, we show that $n$ qubits of communication cannot achieve this.
\end{proof}
The above lemma provides the optimal method for implementing $n$ \emph{CZ} gates in parallel while exchanging $n+1$ qubits. Such a protocol would prove useful for quantum computing architectures where quantum communication is used to interface and implement gates between clusters of highly controllable qubits. As an example, in quantum error correction one could imagine using the Steane code \cite{steane1996multiple} to protect 2 logical qubits using 2 spatially separated clusters of 7 physical qubits. To implement a \emph{CZ} gate between the logical qubits requires 7 \emph{CZ}s to be performed in parallel between the physical qubits.

Our protocol achieves this while exchanging only 8 qubits whereas the naive protocol would send 14 qubits. Protocols based solely on shared entanglement and classical communication \cite{gottesman1998heisenberg, eisert2000optimal, collins2001nonlocal} use 7 pairs of ebits, 14 bits of communication and the implementation of 14 measurements while their coherent counterpart \cite{harrow2004coherent} requires 1 shared ebit and 8 qubits of communication.

In Appendix \ref{ap:clean quant IP} we give a clean quantum protocol for computing $\textit{IP}_n$:
\begin{theorem} \label{th:clean quant IP}
The clean, quantum communication complexity of exactly computing $\textit{IP}_n$ satisfies:
\begin{equation}
n+1\leq \textit{Q}_{\textit{clean}}\left(\textit{IP}_n\right)\leq \begin{cases}
n+3 & \text{for }n \text{ odd},\\
n+2 & \text{for }n \text{ even}.
\end{cases}
\end{equation}
No ancilla qubits are required.
\end{theorem}
By adapting the protocol from Lemma~\ref{le:IP in phase} we can also show that $\textit{IP}_n$ can be computed cleanly using 2 qubits and $n+1$ bits of classical communication. We give this protocol in Appendix \ref{ap:clean mix IP}.

Our novel quantum communication protocols inspire a classical protocol for Inner Product (given in Appendix \ref{ap:clean class IP}) which is near optimal and for which only the naive $2n$ protocol was known before:
\begin{theorem} \label{th:clean class IP}
The clean, classical communication complexity of exactly computing $\textit{IP}_n$ satisfies:
\begin{equation}
n+1\leq\textit{C}_{\textit{clean}}\left(\textit{IP}_n\right)\leq n + 4\sqrt{n}+\frac{1}{\sqrt{n}-1} +2.
\end{equation}
No ancilla bits are required.
\end{theorem}

\emph{Generic functions.} In contrast to Theorem \ref{th:clean class IP}, we will show that nearly all Boolean functions on $n$-bit inputs require $2n-O\left(\log n\right)$ bits of classical communication to compute cleanly. The proof follows from the following two lemmas. In what follows, $X$ and $Y$ are the random variables for Alice and Bob's inputs and $A$ and $B$ are the random variables received by Alice and Bob respectively through the communication that takes place over the course of the protocol. By $|a|$ we denote the number of bits received by Alice and $|b|$ the number of bits received by Bob.
\begin{lemma} \label{le:Kol comp}
Consider picking uniformly at random a Boolean function $f_n$ on $n$-bit inputs. Then with probability $1-o(1)$, all protocols that compute $f_n$ exactly are such that either:
\begin{enumerate}
\item Alice must receive:
\begin{equation} \label{eq:|a| bound}
|a|\geq n-\log\left(n+1\right)-2,
\end{equation}
bits and there exists a uniform distribution over at least half the pairs of inputs such that:
\begin{equation} \label{eq:I(Y|AX) bound}
I\left(Y:AX\right)\geq n -\log\left(n+1\right)-3.
\end{equation}
\end{enumerate}
Or:
\begin{enumerate}
\setcounter{enumi}{1}
\item Bob must receive:
\begin{equation} \label{eq:|b| bound}
|b|\geq n-\log\left(n+1\right)-2,
\end{equation}
bits and there exists a uniform distribution over at least half the pairs of inputs such that:
\begin{equation}
I\left(X:BY\right)\geq n -\log\left(n+1\right)-3.
\end{equation}
\end{enumerate}
\end{lemma}
\begin{proof}
The full proof is given in Appendix~\ref{ap:clean class}. To prove the first two bounds, begin by noting that the communication matrix $M^f$ (defined by $M^f_{xy}=f\left(x,y\right)$) of a random Boolean function has large \emph{Kolmogorov complexity} with high probability. However, a classical protocol for computing $f$ partitions the matrix into \emph{rectangles} (see Appendix \ref{ap:com comp}), each of which has low Kolmogorov complexity. If one of these rectangles is large enough (which happens when the amount of communication that takes place in one direction is small), then the Kolmogorov complexity of $M^f$ will also be low. Such an $M^f$ is shown in Figure~\ref{fig:ComplexRec}a. Comparing these two statements leads to the bounds on $\left|a\right|$ and $\left|b\right|$.

\begin{figure}
\centering
\includegraphics[width=0.9\linewidth]{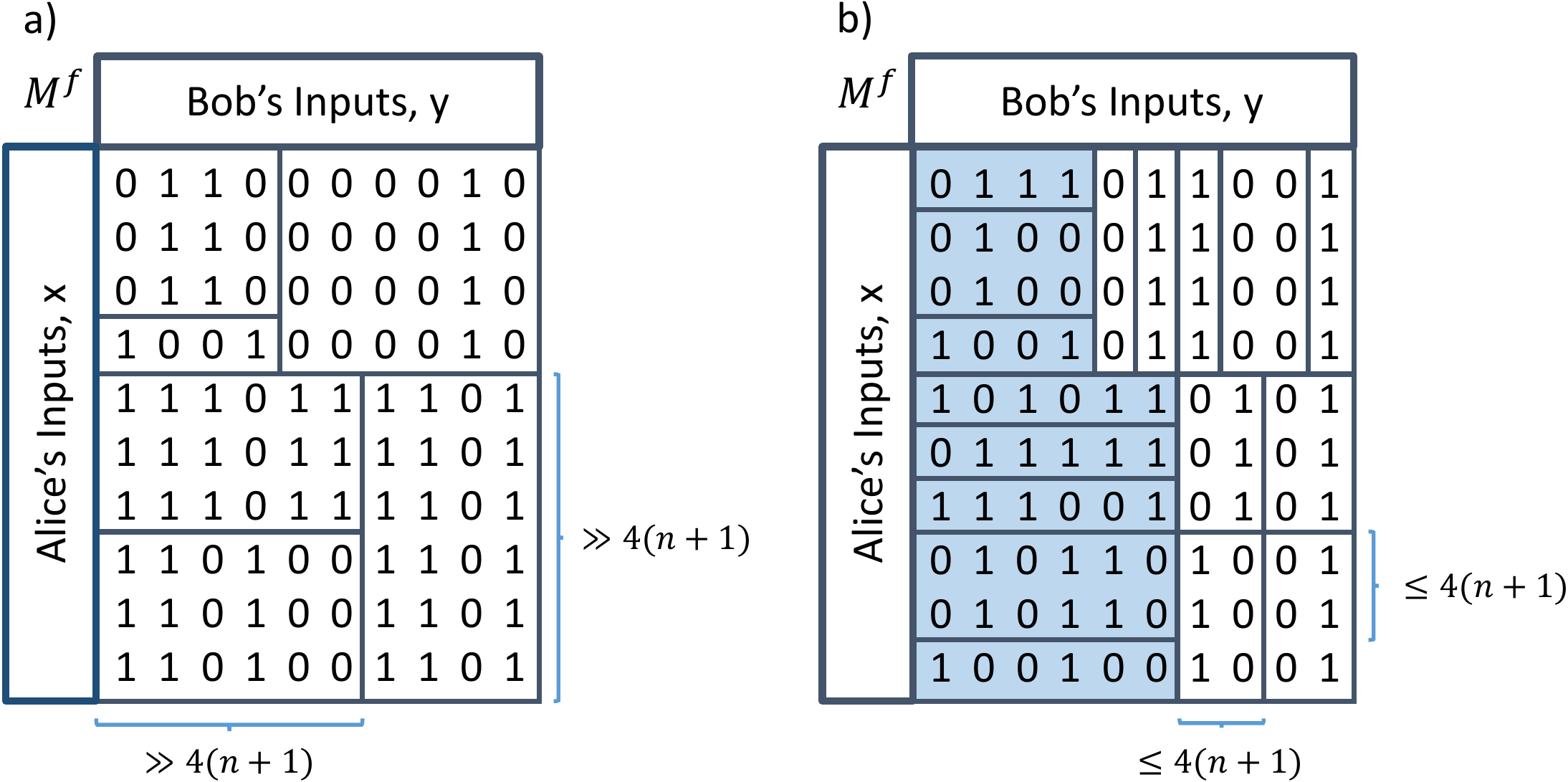}
\caption{\emph{Partitions of the communication matrix into rectangles.} Note that knowledge of $y$, together with knowledge of which rectangle the players' input pair belongs to, allows Bob to correctly deduce the value of $f\left(x,y\right)$. a) As there exists a protocol for computing $f$ that partitions $M^f$ into large rectangles, the Kolmogorov complexity of $M^f$ is low. b) For $M^f$ to have high Kolmogorov complexity, all protocols for computing $f$ must partition $M^f$ into either very narrow or very thin rectangles. 
To produce the bound in Eq.~\eqref{eq:I(Y|AX) bound}, we take a distribution over the shaded rectangles.}
\label{fig:ComplexRec}
\end{figure}

These bounds imply that the rectangles induced by any protocol for computing most $f_n$ must either be very short or very thin as shown in Figure \ref{fig:ComplexRec}b. In fact, they cannot be larger than $4\left(n+1\right)\times2^n$ nor $2^n\times 4\left(n+1\right)$. Either at least half the inputs will belong to very short rectangles or at least half the inputs will belong to very thin ones. By taking a distribution over the larger set, we induce a direction into the communication that occurs in the protocol to ensure that one of Eqs.~\eqref{eq:|a| bound} and \eqref{eq:|b| bound} holds and bound the related mutual information. For example, consider the case where more than half the input pairs lie in rectangles of size less than $2^n\times 4\left(n+1\right)$ (as shown in the figure) and the distribution over $x$ and $y$ is formed by picking Alice and Bob's inputs uniformly at random from such rectangles. Then, at the end of the protocol, Alice will know that Bob received one of at most $4\left(n+1\right)$ inputs and Eq.~\eqref{eq:|a| bound} will hold. Hence:
\begin{equation*}
I\left(Y:AX\right)=H\left(Y\right)-H\left(Y|AX\right)\geq n -\log\left(n+1\right)-3,
\end{equation*}
as required.
\end{proof}

The previous lemma indicates that to compute most functions, either Alice or Bob must receive close to the entirety of the other player's input. In the next lemma we shall see that a similar amount of information (and hence communication) must flow back in the other direction to make the protocol clean.
\begin{lemma} \label{le:Clean leak}
Let $f$ be a Boolean function and its inputs be chosen according to some distribution. Then, in a clean protocol for exactly computing~$f$:
\begin{equation} 
|b|\geq I\left(Y:XA\right)-I\left(X:Y\right),
\end{equation}
and:
\begin{equation}\label{eq:X:YB bound}
|a|\geq I\left(X:YB\right)-I\left(X:Y\right)-1.
\end{equation}
\end{lemma}
\begin{proof}
The full proof can be found in Appendix~\ref{ap:clean class}. It revolves around considering a protocol as $r$ rounds in which each player speaks. A schematic of an individual round is shown in Figure \ref{fig:schematic}.
\begin{figure}
\centering
\includegraphics[width=0.9\linewidth]{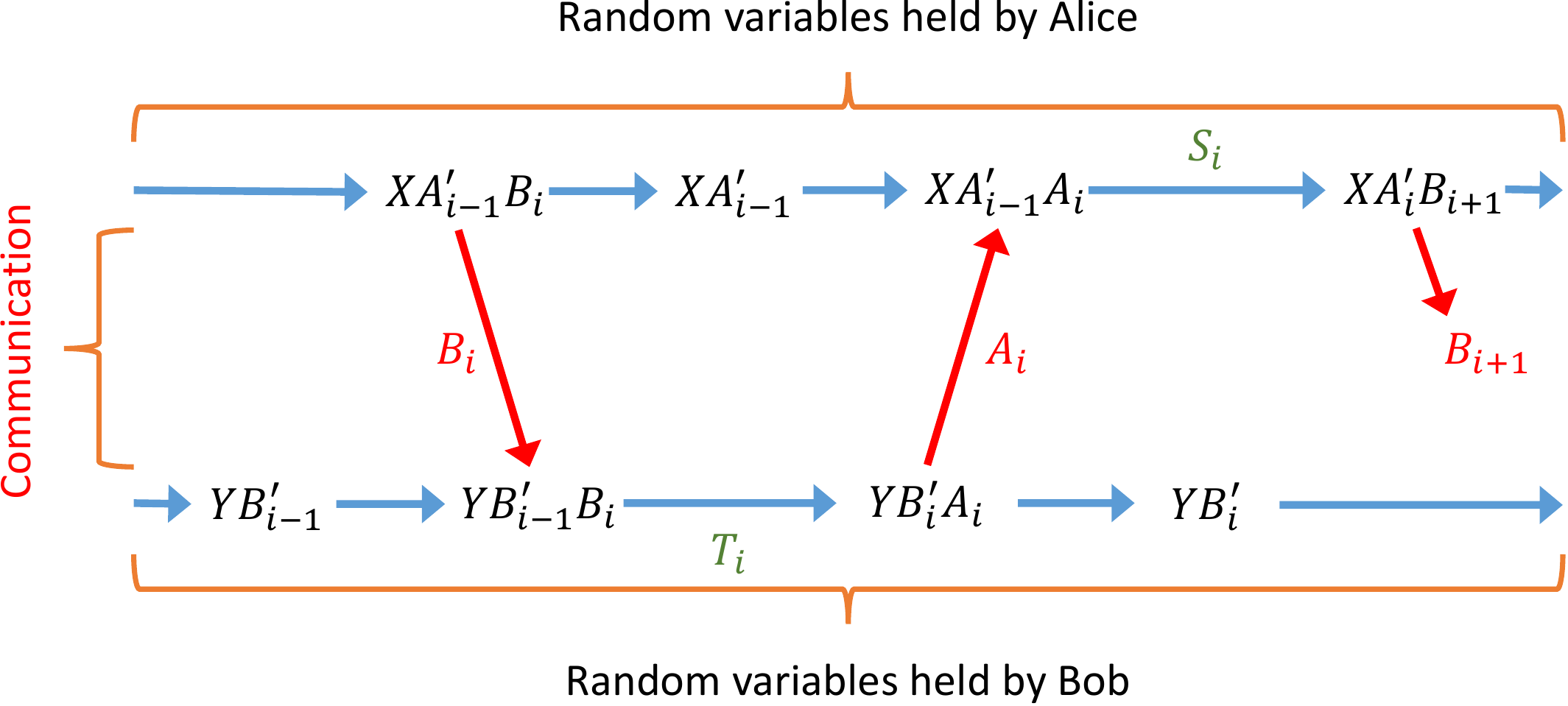}
\caption{\emph{Schematic of a classical communication protocol.} Here we show how the random variables held by each player change during round $i$ of a communication protocol. Primed variables denote local memories while non-primed variables are communication. Each player uses a deterministic, reversible function ($S_i$ and $T_i$) to determine their next message and update their local memory.}
\label{fig:schematic}
\end{figure}
The bounds are then constructed by noting that in each round the players' messages are produced by a deterministic, reversible function of their inputs, local memory (denoted by $A'_i$ and $B'_i$) and the last message received. To obtain (for example) Eq.~\eqref{eq:X:YB bound}, the chain rule for the conditional mutual information can then be used to write:
\begin{align*}
I\left(X:YB\right)=&I\left(X:Y\right) + I\left(X:B'_r|Y\right)\\
& + \sum_{i=1}^{r} I\left(X:A_i|YB'_i B_{i+1}\dots B_r\right)\\
\leq& I\left(X:Y\right)+1+|a|,
\end{align*}
where in the last line we have used the fact that that the protocol is clean and that the conditional mutual information can be upper bounded by the number of bits contained in $A_i$.
\end{proof}

Combining these two lemmas, together with the fact that $I\left(X:Y\right)\leq1$ for uniform distributions over at least half the possible inputs, we obtain our result:
\begin{theorem} \label{th:rand clas}
Consider exactly computing a Boolean function $f_n$ on $n$-bit inputs that has been picked uniformly at random. Then with probability $1-o(1)$:
\begin{equation}
\textit{C}_{\textit{clean}}\left(f_n\right)\geq 2n-2\log\left( n+1\right)-7.
\end{equation}
\end{theorem}

In the case of quantum protocols, a similar result holds in the entanglement assisted case. Proving this result (Appendix \ref{ap:rand ent}) makes use of the fully quantum notion of information complexity introduced recently in \cite{touchette2015quantum}. The proof follows a similar structure to the classical result: arguing that for most functions close to $n$ bits of information has to flow from Alice to Bob and for the protocol to be clean an equivalent amount of information has to be returned.
\begin{theorem} \label{th:rand ent}
Consider exactly computing a Boolean function $f_n$ on $n$-bit inputs that has been picked uniformly at random. Then with probability $1-o(1)$:
\begin{equation}
\textit{Q}^*_{\textit{clean}}\left(f_n\right)\geq n-\log n.
\end{equation}
\end{theorem}

\emph{Conclusion.} In this paper we have initiated the study of how big an overhead in communication cost cleanliness requires. For the Inner Product function (and the task of implementing $n$ \emph{CZ} gates in parallel) we have exhibited quantum and classical protocols for which the overhead is low. For most Boolean functions however, we have shown that the additional cost incurred by demanding cleanliness is close to maximal for the classical and entanglement assisted complexities. Many questions remain.

For example, what are the clean, classical and quantum complexity of other notable functions such as Equality and Disjointness? More generally, note that any Boolean function on inputs of length $n$ can be written in the form:
\begin{equation}
f\left(x,y\right)=\sum_{i=1}^{k} P_i\left(x\right)\cdot Q_i\left(y\right),
\end{equation}
where $\left\{P_i\right\}_{i=1}^{k}$ and $\left\{Q_i\right\}_{i=1}^{k}$ are sets of Boolean functions and $k$ is at most $2^n$ \cite{van1999nonlocality}. It follows that for those functions for which such a decomposition exists with small enough $k$, the protocols used in Theorems \ref{th:clean quant IP} and \ref{th:clean class IP} can be used to give non-trivial upper bounds on the quantum and classical clean communication complexities respectively. In particular, this holds for $k<2n-3$ in the quantum case and $k<2n-4\sqrt{2n}+4$ in the classical setting.

As Theorems \ref{th:rand clas} and \ref{th:rand ent} show that the clean, classical and entanglement assisted communication complexity for most functions on $n$ bit inputs is close to maximal, one can ask: does something similar hold for $Q_{\textit{clean}}\left(f\right)$? We leave this as an open question but conjecture it to be close to $2n$ as Inner Product appears somewhat special in its ability to reuse a single ebit efficiently. However, the concept of information cost is somewhat blind to the sending of ebits so the technique used for the entanglement assisted case does not immediately generalize to proving a bound potentially larger than $n$.

\begin{acknowledgments}
We thank Aram Harrow for helpful discussions and for bringing \cite{harrow2004coherent} to our attention. HB was partially funded by the European Commission, through the SIQS project and by the Netherlands Organisation for Scientific Research (NWO) through gravitation grant Networks. MC and CP acknowledge financial support from the European Research Council (ERC Grant Agreement no 337603), the Danish Council for Independent Research (Sapere Aude), VILLUM FONDEN via the QMATH Centre of Excellence (Grant No. 10059) and the Swiss National Science Foundation (project no PP00P2\_150734). JZ is supported by NWO through the research programme 617.023.116 and by the European Commission through the SIQS project.
\end{acknowledgments}

\bibliography{References}
\bibliographystyle{apsrev4-1}

\clearpage
\widetext
\appendix

\section{Preliminaries}

In this Appendix we provide background materials from information theory, communication complexity, the study of Kolmogorov complexity and the concept of information complexity that have been used to prove our results.

\subsection{Information theory}

For a more thorough introduction to the quantities discussed here, see, for example, \cite{wilde2013quantum,nielsen2010quantum}.

\subsubsection{Classical}

To prove Lemmas~\ref{le:Kol comp} and \ref{le:Clean leak}, we need to define the classical mutual information and its conditional analogue. To do this, we first define the following quantities:
\begin{definition}{\textbf{Shannon Entropy.}}
\begin{itemize}
\item Given a random variable $X$, its Shannon entropy is defined by:
\begin{equation}
H\left(X\right)=-\sum_x p\left(x\right)\log p\left(x\right).
\end{equation}
If $X$ has support on $n$ elements, then $H\left(X\right)\leq \log n$.
\item For two random variables $X$ and $Y$, the entropy of $X$ conditioned on knowing $Y$ (the conditional entropy of $X$ given $Y$) is given by:
\begin{equation}
H\left(X|Y\right)=\sum_y p\left(y\right) H\left(X|Y=y\right).
\end{equation}
\end{itemize}
\end{definition}

With these in place, the mutual information is defined as follows:
\begin{definition}{\textbf{Classical mutual information.}}
\begin{itemize}
\item The mutual information between two random variables $X$ and $Y$ is given by:
\begin{equation}
I\left(X:Y\right)=H\left(X\right)-H\left(X|Y\right).
\end{equation}
\item The mutual information between two random variables $X$ and $Y$ conditioned on knowing a third random variable $Z$ (the conditional mutual information) is given by:
\begin{equation}
I\left(X:Y|Z\right)=H\left(X|Z\right)-H\left(X|YZ\right).
\end{equation} 
\end{itemize}
Note that:
\begin{align*}
I\left(X:YZW\right)&=H\left(X\right)-H\left(X|YZW\right)\\
&=H\left(X\right)-H\left(X|YZ\right)+H\left(X|YZ\right)-H\left(X|YZW\right)\\
&=I\left(X:YZ\right)+I\left(X:W|YZ\right).
\end{align*}
\end{definition}

\subsubsection{Quantum}

To define the concept of information complexity and prove Theorem~\ref{th:rand ent} we will also require their quantum analogues:
\begin{definition}{\textbf{von Neumann entropy.}}
Given a quantum state $\rho$ belonging to a Hilbert space $\mathcal{H}$, its von Neumann entropy is defined by:
\begin{equation}
S\left(\rho\right)=-\tr\left[\rho\log\rho\right].
\end{equation}
If $\rho$ is the maximally mixed state on a Hilbert space of dimension $n$, then $S\left(\rho\right)=\log n$.
\end{definition}

\begin{definition}{\textbf{Quantum mutual information.}}
\begin{itemize}
\item Given a composite quantum state $\rho_{AB}$ on a product Hilbert space $\mathcal{H}_{AB}=\mathcal{H}_A\otimes\mathcal{H}_B$, the quantum mutual information between the two components $A$ and $B$ is given by:
\begin{equation}
I\left(A:B\right)=S\left(\rho_A\right)+S\left(\rho_B\right)-S\left(\rho_{AB}\right),
\end{equation}
with the reduced density matrices $\rho_A$ and $\rho_B$ defined by $\rho_A=\tr_B\left[\rho_{AB}\right]$ and $\rho_B=\tr_A\left[\rho_{AB}\right]$.

\item Given a composite quantum state $\rho_{ABC}$ on a product Hilbert space $\mathcal{H}_{ABC}=\mathcal{H}_A\otimes\mathcal{H}_B\otimes\mathcal{H}_C$, the quantum mutual information between two components $A$ and $B$ conditioned on the third component $C$ is given by:
\begin{equation}
I\left(A:B|C\right)=S\left(\rho_{AC}\right)+S\left(\rho_{BC}\right)-S\left(\rho_C\right)-S\left(\rho_{ABC}\right),
\end{equation}
with the reduced density matrices defined in a similar fashion to the above.
\end{itemize}
\end{definition}

\subsection{Communication complexity} \label{ap:com comp}

For a comprehensive introduction to the field of communication complexity, see \cite{kushilevitz1997communication}.

To prove Lemma~\ref{le:Kol comp} we will need the following basic concepts from the theory of communication complexity:
\begin{definition}{\textbf{Communication matrix.}}
Given a Boolean function $f:\left\{0,1\right\}^n\times\left\{0,1\right\}^n\rightarrow\left\{0,1\right\}$, the associated communication matrix $M^f$ is a $2^n\times2^n$ matrix such that:
\begin{equation}
M^f_{x,y}=f\left(x,y\right).
\end{equation}
\end{definition}
\begin{definition}{\textbf{Monochromatic rectangle}}
Given two sets $\mathcal{X}$ and $\mathcal{Y}$, a rectangle is a set $\mathcal{R}=\mathcal{S}\times\mathcal{T}$ where $\mathcal{S}\subseteq\mathcal{X}$ and $\mathcal{T}\subseteq\mathcal{Y}$. Given a function $f$ with domain $\mathcal{X}\times\mathcal{Y}$, a rectangle $\mathcal{R}$ is said to be $f$-monochromatic (shortened to monochromatic) if there exists a constant $z$ such that $f\left(x,y\right)=z$ for all $\left(x,y\right)\in\mathcal{R}$.
\end{definition}
Part of the relevance of rectangles to communication tasks is captured in the following lemma:
\begin{lemma}
Any classical protocol for computing a function $f$ (such that both players learn the answer), partitions the communication matrix $M^f$ into monochromatic rectangles. 
\end{lemma}
\begin{proof}
See, for example, \cite[Lemma 1.16]{kushilevitz1997communication}.
\end{proof}

\subsection{Kolmogorov complexity}

The proof of Lemma~\ref{le:Kol comp} will also make use of the concept of Kolmogorov complexity.  A concise introduction to the topic can be found in \cite{trevisan2015notes}. For more detail see \cite{li2013introduction} or \cite[Chapter 14]{cover2012elements}.
\begin{definition}{\textbf{Kolmogorov complexity (informal).}}
Given a universal computer $\mathcal{U}$, the Kolmogorov complexity of an $n$-bit string $s$ with respect to $\mathcal{U}$, $K_{\mathcal{U}}\left(s\right)$, is defined to be the length of the shortest program that when implemented on $\mathcal{U}$ prints $s$ and then halts.

The conditional Kolmogorov complexity of $s$ given knowledge of $n$, $K_{\mathcal{U}}\left(s|n\right)$ is the shortest program  length when $\mathcal{U}$ has the value of $n$ made freely available to it.
\end{definition}
The choice of universal computer impacts upon the Kolmogorov complexity by at most an additive constant and hence we shall drop the subscript $\mathcal{U}$ and take the Kolmogorov complexity to be defined with respect to some fixed universal computer. More formal definitions and additional details can be found in the references.

In what follows, we shall make use of the following lemma:
\begin{lemma} \label{le:rand Kol}
For every $n$ and every $c$, the probability that a $n$-bit string $s$, chosen uniformly at random, is such that:
\begin{equation}
K\left(s|n\right)\geq n-c,
\end{equation}
is greater than $1-2^{-c}$. 
\end{lemma}
\begin{proof}
See, for example, \cite[Fact 2]{trevisan2015notes} or \cite[Theorem 14.5.1]{cover2012elements}.
\end{proof}

\subsection{Quantum information cost} \label{ap:inf comp}

While the communication complexity measures the number of physical bits or qubits that Alice and Bob exchange during the course of a protocol, the information cost seeks to capture the amount of information the players reveal regarding their inputs. As such, it will depend on the distribution that the players' inputs are drawn according to (to see this note that if Alice sends her entire input to Bob, then if their inputs are perfectly correlated, Alice's message reveals nothing to Bob. Alternatively, if the inputs are not perfectly correlated then Alice will send some information). The information cost of an entanglement assisted, quantum protocol was defined recently in \cite{touchette2015quantum} to be:
\begin{definition}{\textbf{Quantum information cost.}}
Let $\Pi$ be a quantum protocol applied to an input state $\rho_{A_{\textrm{IN}}B_{\textrm{IN}}}$. Let $\ket{\Psi}_{A_{\textrm{IN}}B_{\textrm{IN}}R}$ be the purification of $\rho_{A_{\textrm{IN}}B_{\textrm{IN}}}$. The quantum information cost of $\Pi$ applied to $\rho_{A_{\textrm{IN}}B_{\textrm{IN}}}$ is given by:
\begin{equation}
\textit{QIC}\left(\Pi,\ket{\Psi}_{A_{\textrm{IN}}B_{\textrm{IN}}R}\right)=\frac{1}{2}\sum_{i\textrm{ odd}} I\left(C_i:R|B_{i-1}\right) +\frac{1}{2} \sum_{i\textrm{ even}} I\left(C_i:R|A_{i-1}\right),
\end{equation}
where the systems $A_i$, $B_i$, $C_i$ and $R$ are defined in Figure \ref{fig:InfCost}.

\begin{figure}
\centering
\includegraphics[width=.9\linewidth]{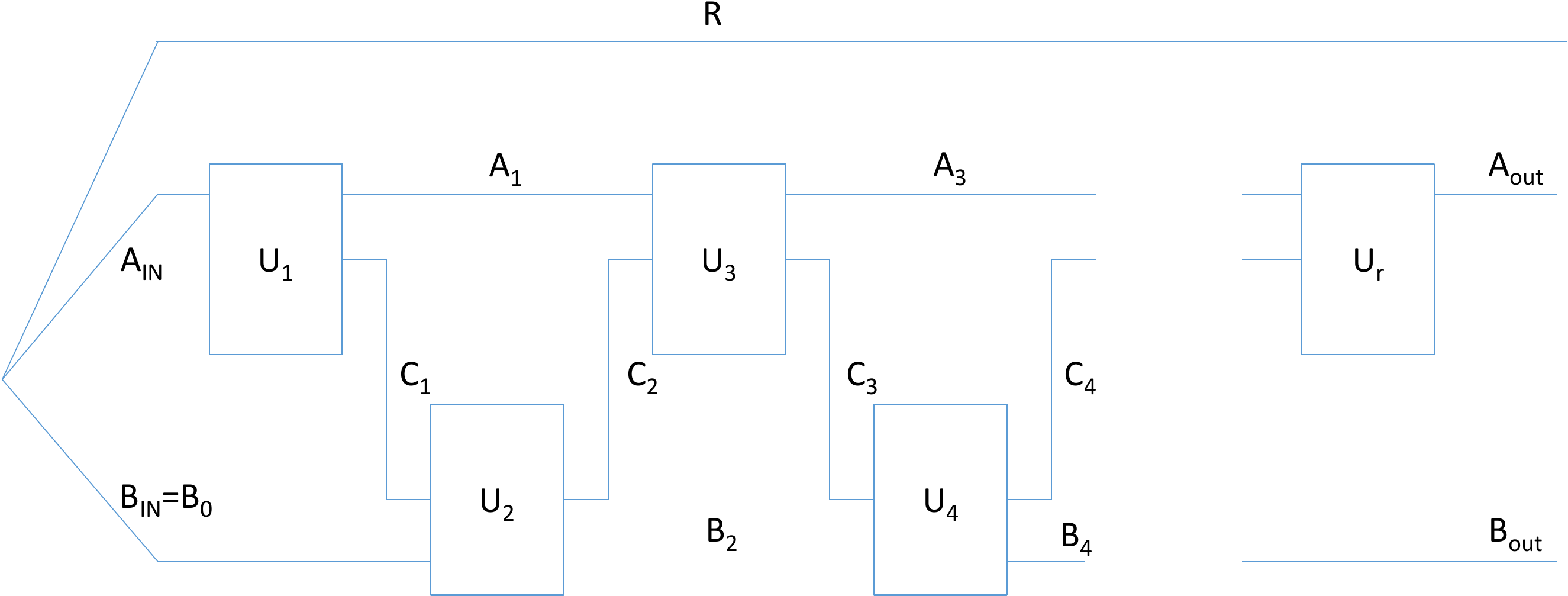}
\caption{\emph{Schematic diagram of the systems involved in a quantum communication protocol.} In each round of the protocol a player applies a unitary to the qubits in their possession. This determines the state of the message system they send to the other player and updates their local state-space. The number of qubits exchanged during the protocol is $\sum_{i}\left\lceil\log\textrm{dim}\left(C_i\right)\right\rceil$. The system $R$ holds the purification of the players' input state.}
\label{fig:InfCost}
\end{figure}

The information that leaks from Alice to Bob during the protocol shall be denoted $\textit{QLA}\left(\Pi,\ket{\Psi}_{A_{\textrm{IN}}B_{\textrm{IN}}R}\right)$ and is given by:
\begin{equation}
\textit{QLA}\left(\Pi,\ket{\Psi}_{A_{\textrm{IN}}B_{\textrm{IN}}R}\right)=\sum_{i\textrm{ odd}} I\left(C_i:R|B_{i-1}\right).
\end{equation}
The information that leaks from Bob to Alice during the protocol can be defined in a similar way.
\end{definition}
Note that as $\log\textrm{dim}\left(C_i\right)\geq\frac{1}{2}I\left(C_i:R|B_{i-1}\right)$ for odd $i$ and $\log\textrm{dim}\left(C_i\right)\geq\frac{1}{2}I\left(C_i:R|A_{i-1}\right)$ for even $i$, the information cost of a protocol on any input state provides a lower bound on the number of qubits exchanged during the protocol \cite{touchette2015quantum}.

\newpage

\section{Clean protocols for Inner Product}

In this Appendix we give the full, explicit protocols for cleanly computing: $\textit{IP}_n$ in the phase using $n+2$ qubits and without ancilla qubits, $\textit{IP}_n$ using $n+3$ qubits (Theorem \ref{th:clean quant IP}), $\textit{IP}_n$ using $n+1$ bits and 2 qubits and $\textit{IP}_n$ using $n+O\left(\sqrt{n}\right)$ bits (Theorem \ref{th:clean class IP}).

\subsection{Quantum protocols for computing Inner Product in the phase.} \label{ap:IP in phase}

Here we show how to compute $\textit{IP}_n$ in the phase, cleanly and without prior entanglement, and without the need for any ancilla qubits.

\begin{lemma} \label{le:IP in phase2}
The clean, quantum communication complexity of exactly implementing $\textit{IP}_n$ in the phase without ancilla qubits satisfies:
\begin{equation}
\textit{Q}_{\textit{clean}}\left(\textit{IP}_{n}^{\textit{phase}}\right)\leq \begin{cases}
n+1 & \text{for }n \text{ odd},\\
n+2 & \text{for }n \text{ even}.
\end{cases}
\end{equation}
\end{lemma}
\begin{proof}
Suppose Alice starts with the state $\ket{x}=\ket{x_1}\ket{x_2}\dots\ket{x_n}$ in qubits labeled $A_1,\dots,A_n$ and Bob starts with $\ket{y}=\ket{y_1}\ket{y_2}\dots\ket{y_n}$ in qubits labeled $B_1,\dots,B_n$. For simplicity of exposition, we shall assume $n$ is even.
\begin{enumerate}
\item For each even $i$, Alice applies a \emph{CZ} gate between $A_1$ and $A_i$. This applies a global phase of $\left(-1\right)^{x_1\cdot\sum_{i:\textrm{even}}x_i}$. Alice sends qubit $A_1$ to Bob.
\item For each odd $i$, Bob applies a \emph{CZ} between $A_1$ and $B_i$. This applies a global phase of $\left(-1\right)^{x_1\cdot\sum_{i:\textrm{odd}}y_i}$.
\item Bob performs a \emph{CNOT} gate on $A_1$ using $B_2$ as the control qubit. This leaves $A_1$ in the state $\ket{x_1\oplus y_2}$. Bob sends $A_1$ to Alice.
\item In round $j$ of the protocol ($2\leq j \leq \frac{n}{2}$):
\begin{enumerate}
\item Alice performs a \emph{CZ} between $A_1$ and $A_{\left(2j-2\right)}$. This applies a global phase of $\left(-1\right)^{x_1\cdot x_{\left(2j-2\right)}\oplus y_{\left(2j-2\right)}\cdot x_{\left(2j-2\right)}}$.
\item Alice performs a \emph{CNOT} on $A_1$ using $A_{\left(2j-1\right)}$ as the control qubit. This leaves $A_1$ in the state $\ket{x_1\oplus y_{\left(2j-2\right)}\oplus x_{\left(2j-1\right)}}$. She sends $A_1$ to Bob.
\item Bob performs a \emph{CNOT} on $A_1$ using $B_{\left(2j-2\right)}$ as the control qubit. This leaves $A_1$ in the state $\ket{x_1\oplus x_{\left(2j-1\right)}}$.
\item Bob performs a \emph{CZ} between $A_1$ and $B_{\left(2j-1\right)}$. This applies a global phase of $\left(-1\right)^{x_1\cdot y_{\left(2j-1\right)} \oplus x_{\left(2j-1\right)}\cdot y_{\left(2j-1\right)}}$.
\item Bob performs a \emph{CNOT} on $A_1$ using $B_{2j}$ as the control qubit. This leaves $A_1$ in the state $\ket{x_1\oplus x_{\left(2j-1\right)} \oplus y_{2j}}$. Bob sends $A_1$ to Alice.
\item Alice performs a \emph{CNOT} on $A_1$ using $A_{\left(2j-1\right)}$ as the control qubit. This leaves $A_1$ in the state $\ket{x_1\oplus y_{2j}}$.
\end{enumerate}
\item After round $n/2$, Alice performs a \emph{CZ} between $A_1$ and $A_{n}$. This applies a global phase of $\left(-1\right)^{x_1\cdot x_{n} \oplus y_{n}\cdot x_{n}}$. The overall global phase after this step is now $\left(-1\right)^{x\cdot y}$. Alice then sends $A_1$ back to Bob.
\item Bob applies a \emph{CNOT} on $A_1$ using $B_{n}$ as the control qubit. This leaves $A_1$ in the state $\ket{x_1}$ which he sends back to Alice, completing the protocol.
\end{enumerate}
In total, this protocol sends $n+2$ qubits. For odd $n$ only $n+1$ qubits of communication are required. Here Bob applies the final operation to the global phase and sends $A_1$ in the state $\ket{x_1\oplus x_n}$ to Alice who converts this back into $\ket{x_1}$.
\end{proof}

The above protocol and its counterpart from Lemma~\ref{le:IP in phase}, have a surprising twist to them. If it is applied to a uniform superposition of inputs on Alice's side, it results in Bob's input being sent to Alice:
\begin{equation}
\frac{1}{\sqrt{2^n}}\sum_{x\in\left\{0,1\right\}^n} \ket{x}\ket{y}\stackrel{\textrm{protocol}}{\longmapsto} \ket{y}\ket{y},
\end{equation}
while if Bob inputs a uniform superposition, he obtains Alice's input. At first glance, this seems counter-intuitive: the protocol is capable of sending $n$ bits from either Alice to Bob \emph{or} Bob to Alice with only $n+1$ qubits of communication in total and without the aid of pre-shared entanglement.

Closer examination reveals that this effect is due to superdense coding -- to which the protocol reduces in this case. If Alice runs the protocol in superposition, then in the first step she sends half of an ebit to Bob. Bob's first operations then correspond to encoding the value of $y_1$ and $y_2$ in this ebit as per superdense coding. Once he sends this back to Alice, her next steps correspond to the decoding operation and, if she applies Hadamards to her first two registers, she obtains the value of $y_1$ and $y_2$. The protocol then repeats these steps, resulting in $y$ being transferred to Alice. A similar reduction to superdense coding occurs if Bob runs the protocol in superposition instead.

\subsection{Clean, quantum protocol for Inner Product} \label{ap:clean quant IP}

Here we give the full proof of Theorem \ref{th:clean quant IP}.

\begin{duplicate*}[\textbf{Restatement of Theorem \ref{th:clean quant IP}}]
The clean, quantum communication complexity of exactly computing $\textit{IP}_n$ satisfies:
\begin{equation}
n+1\leq \textit{Q}_{\textit{clean}}\left(\textit{IP}_n\right)\leq \begin{cases}
n+3 & \text{for }n \text{ odd},\\
n+2 & \text{for }n \text{ even}.
\end{cases}
\end{equation}
No ancillary registers are required.
\end{duplicate*}
\begin{proof}
The clean, quantum protocol for achieving the upper bound runs as follows. For simplicity, we assume that $n$ is even:
\begin{enumerate}
\item Initially Alice takes the input registers $A_1$ and $A_2$ containing $\ket{x_1}$ and $\ket{x_2}$ respectively. She sends $A_1$ and $A_2$ to Bob.
\item Bob uses $A_1$ and $A_2$ to cleanly compute the value of $x_1\cdot\sum_{i\textrm{ odd}}y_i+x_2\cdot\sum_{i\textrm{ even}}y_i$ and stores the result in the answer register. This can be done without using any ancillas.
\item Bob applies a Hadamard gate to $A_1$ followed by a \emph{CNOT} to $A_1A_2$ using $A_1$ as the control qubit. This results in the state $\frac{1}{\sqrt{2}}\left(\ket{0}_{A_1}\ket{x_2}_{A_2}+\left(-1\right)^{x_1}\ket{1}_{A_1}\ket{\bar{x}_2}_{A_2}\right)$. Bob sends $A_1$ to Alice.
\item In round $j$ of the protocol ($2\leq j\leq \frac{n}{2}$):
\begin{enumerate}
\item Alice performs a \emph{CZ} gate between $A_1$ and $A_{2j-1}$ followed by a \emph{CNOT} gate to $A_1$ and $A_{2j}$ using $A_{2j}$ as the control qubit. This results in the state $\frac{1}{\sqrt{2}}\left(\ket{x_{2j}}_{A_1}\ket{x_2}_{A_2}+\left(-1\right)^{x_{2j-1}+x_1}\ket{\bar{x}_{2j}}_{A_1}\ket{\bar{x}_2}_{A_2}\right)$. She sends $A_1$ to Bob.
\item Bob performs a \emph{CNOT} on $A_1$ using $A_2$ as the control qubit. He then applies a Hadamard gate to $A_2$. This results in the state $\ket{x_{2j}\oplus x_2}_{A_1}\ket{x_{2j-1}\oplus x_1}_{A_2}$.
\item Bob cleanly computes the inner product between $A_1$ and $B_{2j}$, and $A_2$ and $B_{2j-1}$, storing the result by \emph{XORing} onto the answer register. 
\item Bob applies a Hadamard to $A_2$ before using it as a control qubit to perform a \emph{CNOT} on $A_1$ to recreate the state $\frac{1}{\sqrt{2}}\left(\ket{x_{2j}}_{A_1}\ket{x_2}_{A_2}+\left(-1\right)^{x_{2j-1}+x_1}\ket{\bar{x}_{2j}}_{A_1}\ket{\bar{x}_2}_{A_2}\right)$. He sends $A_1$ back to Alice.
\item Alice performs a \emph{CNOT} gate to $A_1$ and $A_{2j}$ using $A_{2j}$ as the control qubit followed by a \emph{CZ} gate between $A_1$ and $A_{2j-1}$. This recreates the state $\frac{1}{\sqrt{2}}\left(\ket{0}_{A_1}\ket{x_2}_{A_2}+\left(-1\right)^{x_1}\ket{1}_{A_1}\ket{\bar{x}_2}_{A_2}\right)$.
\end{enumerate}
\item After round $n/2$, the answer register holds the value of $x\cdot y$. To complete the protocol, Bob sends $A_2$ back to Alice who performs a \emph{CNOT} gate on $A_1 A_2$ using $A_1$ as the control qubit followed by a Hadamard to $A_1$. This restores $A_1$ and $A_2$ to the state $\ket{x_1}_{A_1}\ket{x_2}_{A_2}$.
\end{enumerate}
In total this protocol sends $n+2$ qubits and requires no ancillas. For odd $n$, $\textit{IP}_n$ can be calculated by running the above protocol on the first $n-1$ input registers before Alice sends $x_n$ to Bob who computes $x_n\cdot y_n$ and sends $x_n$ back to Alice to complete the protocol. Thus $n+3$ qubits of communication and no ancillas are required for odd $n$.

The lower bound follows either from the lower bound in Lemma~\ref{le:IP in phase} or more simply from a result in \cite{buhrman2001communication}. There it was shown that:
\begin{equation}
Q_{\textit{clean}}\left(f\right)\geq\log\textrm{rank}\left(M^f\right)+1.
\end{equation}
As $\textrm{rank}\left(M^{\textit{IP}_n}\right)\geq 2^n-1$ \cite[Example 1.29]{kushilevitz1997communication}, the result follows.
\end{proof}

\subsection{Clean, 2 qubit, $n+1$ bit protocol for Inner Product} \label{ap:clean mix IP}

Here we give a protocol for cleanly computing the Inner Product function based upon the quantum protocol for computing $\textit{IP}_n$ in the phase. It requires only 2 qubits and $n+1$ bits of communication.
\begin{proposition}
$\textit{IP}_n$ can be computed cleanly using:
\begin{equation*}
2\textrm{ qubits and }n+1 \text{ bits}.
\end{equation*}
One ancillary qubit and one ancillary bit is required.
\end{proposition}
\begin{proof}
The protocol in Lemma~\ref{le:IP in phase} can be adapted to compute $\textit{IP}_n$ as follows:
\begin{enumerate}
\item Initially Bob holds the answer register $\ket{z}$ and an additional ancilla in $\ket{0}$. He performs a Hadamard gate on the answer register followed by a \emph{CNOT} on the additional ancilla using the answer register as the control qubit. This results in the entangled state $\frac{1}{\sqrt{2}}\left(\ket{00}+\left(-1\right)^z\ket{11}\right)$. He sends half of this state to Alice (1 qubit of communication).

\item Alice and Bob implement the protocol from Lemma~\ref{le:IP in phase} on $\ket{x}$ and $\ket{y}$ with one change: when they apply the global phases, they condition on their half of the entangled state.

This transforms their entangled pair to the state $\frac{1}{\sqrt{2}}\left(\ket{00}+\left(-1\right)^{z\oplus x\cdot y}\ket{11}\right)$. Alice then sends her half of the entangled state back to Bob (1 qubit of communication).

\item Bob applies a \emph{CNOT} on the ancilla using the answer register as the control. He then applies a Hadamard to the answer register to leave it in the state $\ket{z\oplus x\cdot y}$.
\end{enumerate}
The protocol in Lemma \ref{le:IP in phase} can be implemented using classical communication and 1 ancillary bit and requires $n+1$ bits to be exchanged, so we obtain the result. Only one ancilla qubit need be used to generate the entanglement at the beginning of the protocol.
\end{proof}

\subsection{Clean, classical protocol for Inner Product} \label{ap:clean class IP}

Here we give the full proof of Theorem \ref{th:clean class IP}.

\begin{duplicate*}[\textbf{Restatement of Theorem~\ref{th:clean class IP}}]
The clean, classical communication complexity of  exactly computing $\textit{IP}_n$ satisfies:
\begin{equation}
n+1\leq\textit{C}_{\textit{clean}}\left(\textit{IP}_n\right)\leq n + 4\sqrt{n} +\frac{1}{\sqrt{n}-1} +2.
\end{equation}
\end{duplicate*}
\begin{proof}
We shall construct a clean protocol $\Pi$ that achieves this bound. For fixed block size $k$, define $r=n \textrm{ mod } k$ and define $l$ by writing $n=kl+r$. For $i=1,\dots,l$ we define $x^{\left(i\right)}=x_{\left(i-1\right)k+1}\dots x_{ik}$ to be the $i$th block of $k$ bits of $x$ (and define $y^{\left(i\right)}$ similarly for $y$). By $x^{\left(l+1\right)}=x_{kl+1}\dots x_{n}0\dots0$ we denote the final $r$ bits of $x$ padded with $k-r$ zeros to give it total length $k$ (and again define $y^{\left(l+1\right)}$ similarly for $y$).

In the following we shall assume $l+1$ is even. The clean, classical protocol  $\Pi$ for $\textit{IP}_n$ then runs as follows:
\begin{enumerate}
\item Bob cleanly computes $\sum_{i:\textrm{even}}y^{\left(1\right)}\cdot y^{\left(i\right)}$ mod $2$ and stores it in the answer register $C$. He sends this together with the register (call this $B^{\left(1\right)}$) containing $y^{\left(1\right)}$ to Alice.
\item Alice cleanly computes $\sum_{i:\textrm{odd}}y^{\left(1\right)}\cdot x^{\left(i\right)}$ mod $2$, storing the answer by \emph{XOR}ing it onto $C$. 
\item Alice \emph{XOR}s $x^{\left(2\right)}$ onto the bits in $B^{\left(1\right)}$. This leaves $B^{\left(1\right)}$ in the state $y^{\left(1\right)}\oplus x^{\left(2\right)}$ where the addition is bit-wise and modulo 2.  She sends this together with $C$ back to Bob.
\item In round $j$ of the protocol ($2\leq j\leq \frac{l+1}{2}$):
\begin{enumerate}
\item Bob cleanly computes the inner product (mod 2) of the bits contained in $B^{\left(1\right)}$ and those labeled by $y^{\left(2j-2\right)}$, storing the answer by \emph{XOR}ing it onto $C$.
\item Bob \emph{XOR}s $y^{\left(2j-1\right)}$ onto the bits in $B^{\left(1\right)}$. This leaves $B^{\left(1\right)}$ in the state $y^{\left(1\right)}\oplus x^{\left(2j-2\right)}\oplus y^{\left(2j-1\right)}$. He sends $B^{\left(1\right)}$ and $C$ to Alice.
\item Alice \emph{XOR}s $x^{\left(2j-2\right)}$ onto the bits in $B^{\left(1\right)}$. This leaves $B^{\left(1\right)}$ in the state $y^{\left(1\right)}\oplus y^{\left(2j-1\right)}$.
\item Alice cleanly computes the inner product (mod 2) of the bits contained in $B^{\left(1\right)}$ and those labeled by $x^{\left(2j-1\right)}$, storing the answer by \emph{XOR}ing it onto $C$.
\item Alice \emph{XOR}s $x^{\left(2j\right)}$ onto the bits in $B^{\left(1\right)}$. This leaves $B^{\left(1\right)}$ in the state $y^{\left(1\right)}\oplus y^{\left(2j-1\right)}\oplus x^{\left(2j\right)}$. Alice sends $B^{\left(1\right)}$ and $C$ to Bob.
\item Bob \emph{XOR}s $y^{\left(2j-1\right)}$ onto the bits in $B^{\left(1\right)}$. This leaves $B^{\left(1\right)}$ in the state $y^{\left(1\right)}\oplus x^{\left(2j\right)}$.
\end{enumerate}
\item After round $\frac{l+1}{2}$, Bob holds $B^{\left(1\right)}$ in the state $y^{\left(1\right)}\oplus x^{\left(l+1\right)}$. He cleanly computes the inner product (mod 2) of the bits contained in $B^{\left(1\right)}$ and those labeled by $y^{\left(l+1\right)}$, storing the answer by \emph{XOR}ing it onto $C$. He then sends $B^{\left(1\right)}$ back to Alice and keeps $C$ as it now contains the correct answer.
\item Alice \emph{XOR}s $x^{\left(l+1\right)}$ onto the bits in $B^{\left(1\right)}$. This leaves $B^{\left(1\right)}$ in the state $y^{\left(1\right)}$. She sends $B^{\left(1\right)}$ to Bob, completing the clean protocol.
\end{enumerate}
This protocol exchanges:
\begin{equation}
\textit{C}\left(\Pi\right)=kl+3k+l+1,
\end{equation}
bits of communication. To obtain the bound, it remains to set $k=\left\lfloor \sqrt{n}\right\rfloor$. Then:
\begin{align*}
\textit{C}\left(\Pi\right)&\leq n + 3\left\lfloor \sqrt{n}\right\rfloor + \frac{n}{\left\lfloor \sqrt{n}\right\rfloor} +1\\
&\leq n + 3\sqrt{n} +\frac{n}{\sqrt{n}-1}+1\\
&= n + \frac{4n-3\sqrt{n}}{\sqrt{n}-1}+1\\
&= n + 4\sqrt{n} +\frac{1}{\sqrt{n}-1} +2.
\end{align*}

For odd $l+1$, less than $kl+3k+l+1$ bits of communication is required. Here, Alice applies the final operation to $C$ and sends this, together with $B^{\left(1\right)}$ in the state $y^{\left(1\right)}\oplus y^{\left(l+1\right)}$, back to Bob who converts this back to $y^{\left(1\right)}$ by \emph{XOR}ing $B^{\left(1\right)}$ with $y^{\left(l+1\right)}$. Hence the clean, classical communication complexity of $\textit{IP}_n$ satisfies the claimed bound.

As the local inner products can be computed cleanly using Toffoli gates without ancilla bits, this protocol does not require an ancilla.

The lower bound of $n+1$ follows from the lower bound on the clean, quantum communication complexity given in Theorem \ref{th:clean quant IP}.
\end{proof}

\newpage

\section{Clean complexity for random functions} 

\subsection{Clean classical complexity} \label{ap:clean class}

In this Appendix we give the full proofs of Lemmas~\ref{le:Kol comp} and \ref{le:Clean leak}.

\subsubsection{Proof of Lemma~\ref{le:Kol comp}}

\begin{duplicate*}[\textbf{Restatement of Lemma~\ref{le:Kol comp}}]
Consider picking uniformly at random a Boolean function $f_n$ on $n$-bit inputs. Then with probability $1-o(1)$, all protocols that compute $f_n$ exactly are such that either:
\begin{enumerate}
\item Alice must receive:
\begin{equation}
|a|\geq n-\log\left(n+1\right)-2,
\end{equation}
bits and there exists a uniform distribution over at least half the pairs of inputs such that:
\begin{equation}
I\left(Y:AX\right)\geq n -\log\left(n+1\right)-3.
\end{equation}
\end{enumerate}
Or:
\begin{enumerate}
\setcounter{enumi}{1}
\item Bob must receive:
\begin{equation}
|b|\geq n-\log\left(n+1\right)-2,
\end{equation}
bits and there exists a uniform distribution over at least half the pairs of inputs such that:
\begin{equation}
I\left(X:BY\right)\geq n -\log\left(n+1\right)-3.
\end{equation}
\end{enumerate}
\end{duplicate*}
\begin{proof}
Consider the communication that takes place in a protocol that results in Bob being able to calculate $f_n\left(x,y\right)$ correctly. Let $b$ denote the string of bits sent as messages to Bob during this protocol and $a$ the string of bits sent to Alice. Let $|b|$ and $|a|$ denote the number of bits each player receives over the course of the protocol.

The communication partitions the communication matrix of $f_n\left(x,y\right)$ into rectangles, $\mathcal{R}\left(a,b\right)=\mathcal{X}\left(a,b\right)\times\mathcal{Y}\left(a,b\right)$. Here $\mathcal{X}\left(a,b\right)$ and $\mathcal{Y}\left(a,b\right)$ denote the sets of Alice and Bob's respective inputs that are compatible with the communication that took place. As Bob knows the value of $f_n\left(x,y\right)$ at the end of the protocol, each rectangle will be monochromatically striped. An illustrative example of this is given in Figure~\ref{fig:striped rectangles}.

Note that in a protocol each bit of communication sent between the players partitions the communication matrix. In particular, each bit partitions the rectangles induced by the previous rounds of communication into at most two rectangles. A bit sent from Alice to Bob can split each rectangle horizontally while a bit sent in the other direction divides them vertically. As such, the largest $\mathcal{R}\left(a,b\right)$ has size at least $2^{n-|b|}\times 2^{n-|a|}$. 

\begin{figure}[ht]
\centering
\includegraphics[width=.5\linewidth]{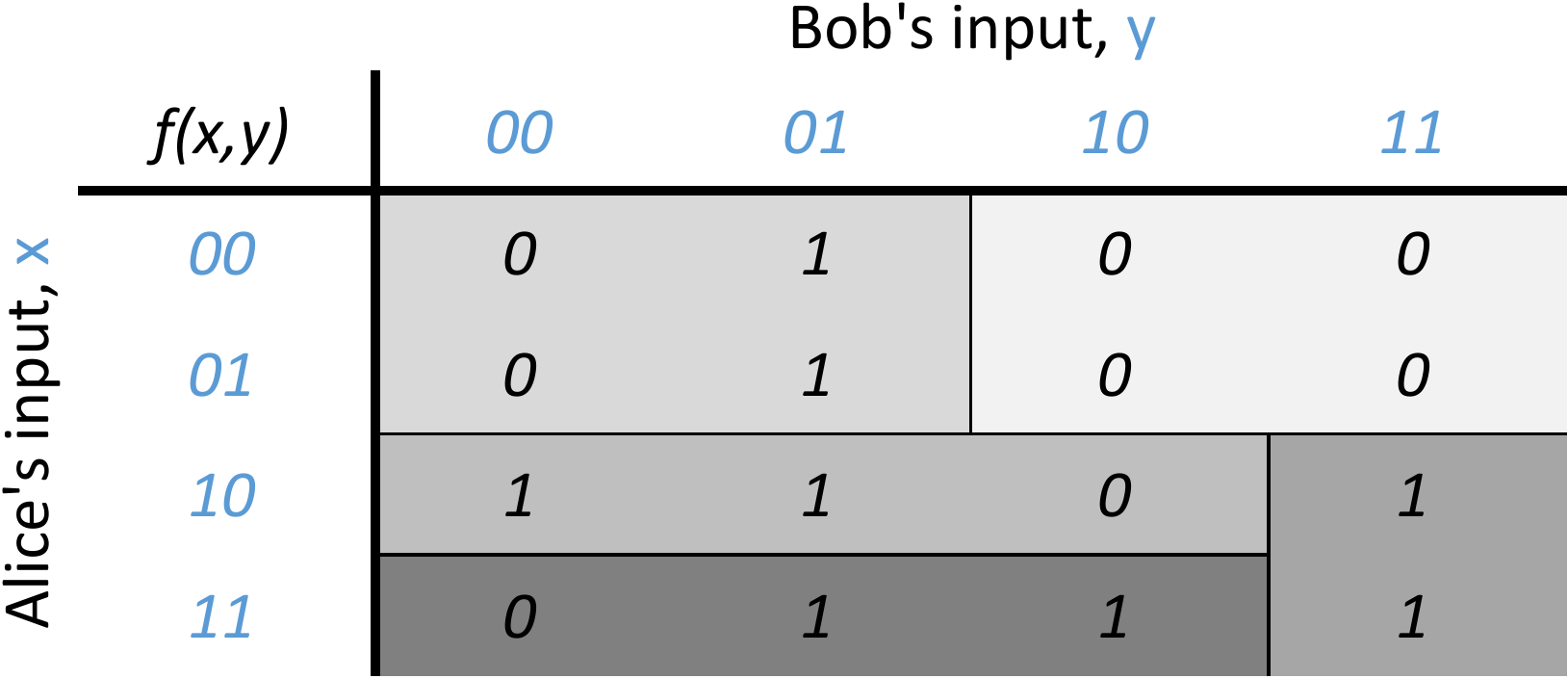}
\caption{\emph{Monochromatically striped rectangles.} Here we show the rectangles generated on $M^f$ by a communication protocol. Knowledge of $y$, together with knowledge of which rectangle the players' input pair belongs to, allows Bob to correctly deduce the value of $f\left(x,y\right)$.}
\label{fig:striped rectangles}.
\end{figure}

We shall now consider the Kolmogorov complexity of the communication matrix of $f_n$. Let $K\left(s|k\right)$ denote the conditional Kolmogorov complexity of a $k$-bit string $s$ given knowledge of $k$. The communication matrix of a bipartite Boolean function $f_n$ on inputs of size $n$ is of size $2^{n}\times2^{n}$. Then, setting $c=\log n$ in Lemma~\ref{le:rand Kol}, the fraction of $f_n$ such that:
\begin{equation} \label{eq:KC lb}
K\left(M^{f_n}|2^{2n}\right)\geq2^{2n}-\log n,
\end{equation}
tends to 1 as $n$ increases.

Now consider the Kolmogorov complexity of a communication matrix which has a monochromatically striped rectangle $\mathcal{R'}=\mathcal{X}'\times\mathcal{Y}'$ of size $2^{n-|b|}\times 2^{n-|a|}$. We can obtain an upper bound on the Kolmogorov complexity as follows. To specify the value of the bits inside $\mathcal{R}'$ we specify the strings inside $\mathcal{X}'$ and $\mathcal{Y}'$. This requires $n2^{n-|b|}$ and $n2^{n-|a|}$ bits respectively together with $2\log n$ bits to specify the value of $|a|$ and $|b|$. To specify the value of each stripe in $\mathcal{R'}$ then requires $2^{n-|a|}$ bits. Finally, we specify the value of the function on each pair of inputs outside of $\mathcal{R'}$. From this we see that for a function with a rectangle of size at least $2^{n-|b|}\times 2^{n-|a|}$, the Kolmogorov complexity of the associated communication matrix satisfies:
\begin{equation}\label{eq:KC ub}
K\left(M^{f_n}|2^{2n}\right)\leq 2^{2n}-2^{2n-|a|-|b|}+n\left(2^{n-|b|}+2^{n-|a|}\right)+2^{n-|a|}+2\log n.
\end{equation}

Comparing Eq.~\eqref{eq:KC ub} with Eq.~\eqref{eq:KC lb}, we see that for a fraction of $f_n$ that tends to 1 with increasing $n$, we have:
\begin{align*}
2^{2n}-\log n &\leq 2^{2n}-2^{2n-|a|-|b|}+n\left(2^{n-|b|}+2^{n-|a|}\right)+2^{n-|a|}+2\log n,\\
\Rightarrow\phantom{{}2^{n}{}}\quad2^{2n-|a|-|b|}&\leq n 2^{n-|b|} +\left(n+1\right)2^{n-|a|} + 3\log n\\
&\leq 2\left(n+1\right) 2^{n-|b|}+2\left(n+1\right) 2^{n-|a|},\\
\Rightarrow\quad\phantom{{}2^{2n-|a|-|b|}{}} 2^{n}&\leq 4\left(n+1\right)2^{\max\left(|a|,|b|\right)}.
\end{align*}
Hence:
\begin{equation}
\max\left(|a|,|b|\right)\geq n-\log\left(n+1\right)-2,
\end{equation}
and for most functions either Alice or Bob must send at least $n-\log\left(n+1\right)-2$ bits of communication to the other player.

We now turn to proving the second half of the lemma. From the first part we know that most $f_n$ have neither monochromatically striped rectangles of size larger than  $4\left(n+1\right)\times 2^n$ nor $2^n \times 4\left(n+1\right)$. Any correct protocol for $f_n$ induces a partition of the communication matrix of $f_n$ into monochromatically striped rectangles and either at least half of all pairs of inputs will lie in rectangles of size smaller than $4\left(n+1\right)\times 2^n$ or in rectangles of size smaller than $2^n \times 4\left(n+1\right)$ (some rectangles will belong to both sets).

Consider the case where at least half of the input pairs lie in rectangles of size no larger than $4\left(n+1\right)\times 2^n$. Consider the distribution over $x$ and $y$ formed by picking Alice and Bob's inputs uniformly at random from pairs belonging to such rectangles. This leads to associated random variables $X,Y,A$ and $B$ for the players' inputs and communication. Now, at the end of the protocol once the players know their inputs belong to a particular rectangle, Bob will know that Alice received one of $4\left(n+1\right)$ inputs and $|b|\geq n-\log\left(n+1\right)-2$ holds. Hence:
\begin{equation}
H\left(X|BY\right)\leq \log\left[4\left(n+1\right)\right]= \log\left(n+1\right)+2.
\end{equation}
Now, using the fact that $H\left(X\right)\geq n-1$ as the distribution is uniform over at least half the inputs, we obtain:
\begin{align*}
I\left(X:BY\right)&=H\left(X\right)-H\left(X|BY\right)\\
&\geq n-\log\left(n+1\right)-3,
\end{align*}
as required.

The second inequality follows similarly when over half of the input pairs lie in rectangles of size no larger than $2^n \times 4\left(n+1\right)$.
\end{proof}

\subsubsection{Proof of Lemma~\ref{le:Clean leak}}

\begin{duplicate*}[\textbf{Restatement of Lemma~\ref{le:Clean leak}}]
Let $f$ be a Boolean function and its inputs be chosen uniformly at random from over at least half of the possible pairs. Then, in a clean protocol for exactly computing $f$:
\begin{equation}  \label{eq:Bob bound}
|b|\geq I\left(Y:XA\right)-I\left(X:Y\right),
\end{equation}
and:
\begin{equation} \label{eq:Alice bound}
|a|\geq I\left(X:YB\right)-I\left(X:Y\right)-1.
\end{equation}
\end{duplicate*}
\begin{proof}
Here we shall show Eq.~\eqref{eq:Alice bound}. The bound in Eq.~\eqref{eq:Bob bound} follows similarly (the difference of 1 bit occurs as only Bob knows the answer at the end of a clean protocol). Let there be $r$ rounds of communication in which both players speak. We shall assume without loss of generality that Alice speaks first. Let $B_i$ denote the bits received by Bob in round $i$ and $B'_i$ his local memory at the end of the round ($B'_0$ is of trivial size). In order to compute his message to Alice, $A_i$, in round $i$, Bob makes a reversible, deterministic transformation:
\begin{equation}
\mathcal{T}_i:YB'_{i-1}B_i\mapsto Y B'_i A_i.
\end{equation}
Note that in a clean protocol, Bob's final register $B'_r$ will contain at most one bit of information regarding $X$, the value of $f\left(x,y\right)$.

Using the chain rule for the conditional mutual information and the fact that reversible, deterministic transformations do not change the entropy, we have:
\begin{align*}
I\left(X:YB\right)&=I\left(X:YB'_0B\right)\\
&=I\left(X:YB'_0B_1\dots B_r\right)\\
&=I\left(X:YB'_1A_1B_2\dots B_r\right)\\
&=I\left(X:YB'_1B_2\dots B_r\right)+I\left(X:A_1|YB'_1B_2\dots B_r\right)\\
&=I\left(X:YB'_r\right)+\sum_{i=1}^{r} I\left(X:A_i|YB'_i B_{i+1}\dots B_r\right)\\
&=I\left(X:Y\right) + I\left(X:B'_r|Y\right) + \sum_{i=1}^{r} I\left(X:A_i|YB'_i B_{i+1}\dots B_r\right)\\
&\leq I\left(X:Y\right)+1+|a|,
\end{align*}
where in the last line we have bounded the terms using the fact that the protocol is clean and that the conditional mutual information can be upper bounded by the number of bits contained in $A_i$. Rearranging this gives:
\begin{equation}
|a|\geq I\left(X:YB\right)-2,
\end{equation}
as claimed.
\end{proof}

\subsection{Clean, entanglement assisted quantum complexity} \label{ap:rand ent}

In this Appendix we prove Theorem~\ref{th:rand ent}:
\begin{duplicate*}[\textbf{Restatement of Theorem~\ref{th:rand ent}}]
Consider exactly computing a Boolean function $f_n$ on $n$-bit inputs that has been picked uniformly at random. Then with probability $1-o(1)$:
\begin{equation}
\textit{Q}^*_{\textit{clean}}\left(f_n\right)\geq n-\log n.
\end{equation}
\end{duplicate*}
To do so we make use of the concept of information complexity defined in Appendix \ref{ap:inf comp}. The following two states will be useful:
\begin{enumerate}
\item When Alice and Bob are both given classical inputs according to some product distribution $\mu_A \times \mu_B$ and the answer register (held by Bob) begins in the state $\ket{-}$, the initial state of the protocol is:
\begin{equation} \label{eq:class-class}
\ket{\Psi_{\textrm{c-c}}}_{A_{\textrm{IN}}B_{\textrm{IN}}R}=\sum_{x,y}{\sqrt{\mu_A\left(x\right)\mu_B\left(y\right)}}\ket{x}_{A}\ket{\vec{0}}_{A_0}\ket{y}_{B}\ket{\vec{0}}_{B_0}\ket{-}_{B_{\textrm{ans}}}\ket{\Phi}_{A_EB_E}\ket{xy}_R,
\end{equation}
where $A_0$ and $B_0$ hold Alice and Bob's local ancilla states and $A_EB_E$ any entanglement initially shared between them.

\item Alice is given a classical input according to $\mu_A$ while Bob is given a quantum superposition over classical inputs according to a distribution $\mu_B$ (together with the answer register again initialized to $\ket{-}$) and entangles it with $n$ blank ancilla qubits using $n$-fold $\emph{CNOT}$ gates. The initial state of the protocol is then:
\begin{equation} \label{eq:class-sup}
\ket{\Psi_{\textrm{c-s}}}_{A_{\textrm{IN}}B_{\textrm{IN}}R}=\sum_{x,y}{\sqrt{\mu_A\left(x\right)\mu_B\left(y\right)}}\ket{x}_{A}\ket{\vec{0}}_{A_0}\ket{y}_{B}\ket{y}_{B'}\ket{\vec{0}}_{B_0}\ket{-}_{B_{\textrm{ans}}}\ket{\Phi}_{A_EB_E}\ket{x}_R.
\end{equation}
The protocol then runs on the $A$, $B$, $A_0$, $B_0$, $A_E$, $B_E$ and $B_{\textrm{ans}}$ registers.
\end{enumerate}
With these defined, a sketch of the proof of the theorem is as follows:
\begin{enumerate}
\item We begin by showing (Corollary~\ref{co:clean}) that for clean protocols on states of the form in Eq.~\eqref{eq:class-class} the amount of information that flows from Alice to Bob must equal the amount of information that flows from Bob to Alice.
\item Next we show (Lemma~\ref{le:cs vs ss}) that for any protocol more information leaks from Alice to Bob if the protocol is applied to the state in Eq.~\eqref{eq:class-class} then if it is applied to the state in Eq.~\eqref{eq:class-sup}.
\item In Theorem~\ref{th:rand clean inf} we use the above two results to argue that for most functions any clean protocol must leak close to $n$ bits of information from Alice to Bob when applied to the state in Eq.~\eqref{eq:class-sup} and hence that for most functions any clean protocol has information cost close to $n$.
\item The fact that the information cost of a protocol lower bounds its communication complexity completes the proof.
\end{enumerate}

We start with the following lemma which relates the amount of information that flows in either direction during the protocol to the information shared between the parties at the beginning and end:
\begin{lemma} \label{le:inf flow}
For any protocol $\Pi$ on any input state $\ket{\Psi}_{A_{\textrm{IN}}B_{\textrm{IN}}R}$:
\begin{equation}
I\left(B_{\textrm{OUT}}:R\right)=I\left(B_{\textrm{IN}}:R\right)-\textit{QLB}\left(\Pi,\ket{\Psi}_{A_{\textrm{IN}}B_{\textrm{IN}}R}\right)+\textit{QLA}\left(\Pi,\ket{\Psi}_{A_{\textrm{IN}}B_{\textrm{IN}}R}\right).
\end{equation}
\end{lemma}
\begin{proof}
Consulting Figure~\ref{fig:InfCost} may aid in following the proof. In particular, it provides an overview of which systems are in existence at any one time. By definition and using the fact that the global state is pure:
\begin{align*}
I\left(B_{\textrm{IN}}:R\right)&=I\left(C_2:R|A_1\right)-I\left(C_2:R|A_1\right)+I\left(B_{\textrm{IN}}:R\right)\\
&=I\left(C_2:R|A_1\right)\\
&\quad+S\left(A_1\right)+S\left(C_2RA_1\right)-S\left(C_2A_1\right)-S\left(RA_1\right)\\
&\quad+S\left(B_0\right)+S\left(R\right)-S\left(B_0R\right)\\
&=I\left(C_2:R|A_1\right)\\
&\quad+S\left(C_1B_0R\right)+S\left(B_2\right)-S\left(B_2R\right)-S\left(C_1B_0\right)\\
&\quad+S\left(B_0\right)+S\left(R\right)-S\left(B_0R\right)\\
&=I\left(C_2:R|A_1\right)+I\left(B_2:R\right)-I\left(C_1:R|B_0\right).
\end{align*}
Iterating this expansion for the (non-conditional) mutual information on the RHS gives:
\begin{equation*}
I\left(B_{\textrm{OUT}}:R\right)=I\left(B_{\textrm{IN}}:R\right)-\sum_{i\textrm{ even}} I\left(C_i:R|A_{i-1}\right)+\sum_{i\textrm{ odd}} I\left(C_i:R|B_{i-1}\right),
\end{equation*}
as required.
\end{proof}
From this, the following corollary follows completing Step 1 of the proof sketch:
\begin{corollary} \label{co:clean}
For any clean protocol $\Pi$ for computing $f$ that is applied to a state $\ket{\Psi_{\textrm{c-c}}}_{A_{\textrm{IN}}B_{\textrm{IN}}R}$ of the form given in Eq.~\eqref{eq:class-class}:
\begin{equation}
\textit{QLB}\left(\Pi,\ket{\Psi_{\textrm{c-c}}}_{A_{\textrm{IN}}B_{\textrm{IN}}R}\right)=\textit{QLA}\left(\Pi,\ket{\Psi_{\textrm{c-c}}}_{A_{\textrm{IN}}B_{\textrm{IN}}R}\right).
\end{equation}
\end{corollary}
\begin{proof}
As the protocol is clean, the output state is:
\begin{equation*}
\ket{\Psi_{\textrm{c-c}}}_{A_{\textrm{OUT}}B_{\textrm{OUT}}R}=\sum_{x,y}{\sqrt{\mu_A\left(x\right)\mu_B\left(y\right)}}\left(-1\right)^{f\left(x,y\right)}\ket{x}_{A}\ket{\vec{0}}_{A_0}\ket{y}_{B}\ket{\vec{0}}_{B_0}\ket{-}_{B_{\textrm{ans}}}\ket{\Phi}_{A_EB_E}\ket{xy}_R,
\end{equation*}
and note that $S\left(A_{\textrm{OUT}}\right)=S\left(A_{\textrm{IN}}\right)$ and $S\left(B_{\textrm{OUT}}\right)=S\left(B_{\textrm{IN}}\right)$. Now, as the total state is pure:
\begin{align*}
I\left(B_{\textrm{OUT}}:R\right)-I\left(B_{\textrm{IN}}:R\right)&=S\left(B_{OUT}\right)+S\left(R\right)-S\left(B_{OUT}R\right)-S\left(B_{IN}\right)-S\left(R\right)+S\left(B_{IN}R\right)\\
&=S\left(B_{OUT}\right)-S\left(A_{OUT}\right)-S\left(B_{\textrm{IN}}\right)+S\left(A_{\textrm{IN}}\right)\\
&=0,
\end{align*}
so $\textit{QLB}\left(\Pi,\ket{\Psi_{\textrm{c-c}}}_{A_{\textrm{IN}}B_{\textrm{IN}}R}\right)=\textit{QLA}\left(\Pi,\ket{\Psi_{\textrm{c-c}}}_{A_{\textrm{IN}}B_{\textrm{IN}}R}\right)$ as required.
\end{proof}
Hence for clean protocols on such states, the amount of information leaked from Alice to Bob is equal to the amount of information leaked from Bob to Alice.

For Step 2 of the proof sketch, the amount of information leaked when the players' inputs are classical-classical or classical-superposition is related as follows. This is a special case of \cite[Theorem 1]{kerenidis2014privacy}.
\begin{lemma} \label{le:cs vs ss}
Given two probability distributions $\mu_A\left(x\right)$ and $\mu_B\left(y\right)$, for any protocol $\Pi$ applied to the appropriate registers:
\begin{equation}
\textit{QLA}\left(\Pi,\ket{\Psi_{\textrm{c-c}}}_{A_{\textrm{IN}}B_{\textrm{IN}}R}\right)\geq \textit{QLA}\left(\Pi,\ket{\Psi_{\textrm{c-s}}}_{A_{\textrm{IN}}B_{\textrm{IN}}R}\right).
\end{equation}
\end{lemma}
\begin{proof}
Considering each term in $\textit{QLA}\left(\Pi,\ket{\Psi_{\textrm{c-s}}}_{A_{\textrm{IN}}B_{\textrm{IN}}R}\right)$ individually:
\begin{align*}
I\left(C_i:R|B_{i-1}B'\right)&=I\left(C_i:R B'|B_{i-1}\right)-I\left(C_i:B'|B_{i-1}\right)\\
&\leq I\left(C_i:R'|B_{i-1}\right),
\end{align*}
where in the last line we have relabeled $RB'$ as $R'$ and used the non-negativity of the conditional mutual information. The system $R'$ is precisely the purifying system in the state $\ket{\Psi_{\textrm{c-c}}}_{A_{\textrm{IN}}B_{\textrm{IN}}R'}$ and hence summing over odd $i$ gives the result.
\end{proof}

We now combine the previous two results to achieve Step 3, showing:
\begin{theorem}\label{th:rand clean inf}
Consider exactly computing a Boolean function $f_n$ on $n$-bit inputs that has been picked uniformly at random. Let $\mu_A\left(x\right)$ and $\mu_B\left(y\right)$ be independent, uniform distributions. Then with probability $1-o\left(1\right)$, any protocol $\Pi$ for computing $f_n$ cleanly is such that for these distributions: 
\begin{equation}
\textit{QIC}\left(\Pi,\ket{\Psi_{\textrm{c-c}}}_{A_{\textrm{IN}}B_{\textrm{IN}}R}\right)\geq n-\log n.
\end{equation}
\end{theorem}
\begin{proof}
Consider running any clean protocol $\Pi$ for $f_n$ on the state $\ket{\Psi_{\textrm{c-s}}}_{A_{\textrm{IN}}B_{\textrm{IN}}R}$ with independent, uniform input distributions. At the end of the protocol, the state will be:
\begin{align*}
\ket{\Psi_{\textrm{c-s}}}_{A_{\textrm{OUT}}B_{\textrm{OUT}}R}&=\sum_x\frac{1}{\sqrt{2^n}}\ket{x}_A\ket{\vec{0}}_{A_0}\left(\sum_y\frac{1}{\sqrt{2^n}}\left(-1\right)^{f_n\left(x,y\right)}\ket{y}_B\ket{y}_{B'}\right)\ket{\vec{0}}_{B_0}\ket{-}_{B_{\textrm{ans}}}\ket{\Phi}_{A_EB_E}\ket{x}_R\\
&\equiv\sum_{x}\frac{1}{\sqrt{2^n}}\ket{x}_A\ket{\vec{0}}_{A_0}\ket{\psi_x}_{BB'}\ket{\vec{0}}_{B_0}\ket{-}_{B_{\textrm{ans}}}\ket{\Phi}_{A_EB_E}\ket{x}_R
\end{align*}
where:
\begin{equation*}
\ket{\psi_x}_{BB'}=\sum_y\frac{1}{\sqrt{2^n}}\left(-1\right)^{f_n\left(x,y\right)}\ket{y}_B\ket{y}_{B'}.
\end{equation*}

The reduced state on systems $B_{\textrm{OUT}}$ and $R$ after the protocol is:
\begin{equation*}
\rho_{B_{\textrm{OUT}}R}=\sum_x \frac{1}{2^n}\ketbra{\psi_x}{\psi_x}_{BB'}\otimes\ketbra{0}{0}_{B_0}\otimes\ketbra{-}{-}_{B_{\textrm{ans}}}\otimes\tr_{A_E}\left[\ketbra{\Phi}{\Phi}_{A_EB_E}\right]\otimes\ketbra{x}{x}_R.
\end{equation*}
Note that $I\left(B_{\textrm{IN}}:R\right)=0$ while:
\begin{align*}
I\left(B_{\textrm{OUT}}:R\right)&=S\left(B_{\textrm{OUT}}\right)+S\left(R\right)-S\left(B_{\textrm{OUT}} R\right)\\
&=S\left(B_{\textrm{OUT}}\right)+S\left(R\right)-S\left(A_{\textrm{OUT}}\right) &&\textrm{(as $\ket{\Psi_{\textrm{c-s}}}_{A_{\textrm{OUT}}B_{\textrm{OUT}}R}$ is pure)}\\
&=S\left(BB'\right)+S\left(B_E\right)+S\left(R\right)-S\left(A\right)-S\left(A_E\right)&&\textrm{(as $\ket{\Phi}_{A_EB_E}$ is product with the other systems)}\\
&=S\left(BB'\right) &&\textrm{(as $\ket{\Psi_{\textrm{c-s}}}_{A_{\textrm{OUT}}B_{\textrm{OUT}}R}$ is symmetric in $A$ and $R$)}.
\end{align*}

Now, using Lemma~\ref{le:inf flow} and the fact that the quantum conditional mutual information (and hence $\textit{QLB}$) is non-negative, we have:
\begin{equation*}
\textit{QLA}\left(\Pi,\ket{\Psi_{\textrm{c-s}}}_{A_{\textrm{IN}}B_{\textrm{IN}}R}\right)\geq S\left(\rho_{BB'}\right),
\end{equation*}
where $\rho_{BB'}=\frac{1}{2^n}\sum_{x}\ketbra{\psi_x}{\psi_x}_{BB'}$. \cite[Theorem IV.1.]{montanaro2007lower} tells us that for a function $f_n$ chosen uniformly at random:
\begin{equation*}
\textrm{Pr}\left[S\left(\rho_{BB'}\right)<\left(1-\delta\right)n\right]\leq e^{-\left(2^{\delta n}-1\right)^2/2}.
\end{equation*}
Setting $\delta=\frac{\log n}{n}$ in the above expression gives us that for nearly all $f_n$:
\begin{equation*}
\textit{QLA}\left(\Pi,\ket{\Psi_{\textrm{c-s}}}_{A_{\textrm{IN}}B_{\textrm{IN}}R}\right)\geq n-\log n.
\end{equation*}

Using Lemma~\ref{le:cs vs ss} we can use this last expression to lower bound $\textit{QLA}\left(\Pi,\ket{\Psi_{\textrm{c-c}}}_{A_{\textrm{IN}}B_{\textrm{IN}}R}\right)$:
\begin{equation*}
\textit{QLA}\left(\Pi,\ket{\Psi_{\textrm{c-c}}}_{A_{\textrm{IN}}B_{\textrm{IN}}R}\right)\geq n-\log n.
\end{equation*}
Finally, using Corollary~\ref{co:clean} gives a lower bound on $\textit{QLB}\left(\Pi,\ket{\Psi_{\textrm{c-c}}}_{A_{\textrm{IN}}B_{\textrm{IN}}R}\right)$:
\begin{equation*}
\textit{QLB}\left(\Pi,\ket{\Psi_{\textrm{c-c}}}_{A_{\textrm{IN}}B_{\textrm{IN}}R}\right)\geq n-\log n,
\end{equation*}
and summing these two expressions gives that for most functions, any clean protocol is such that:
\begin{align*}
\textit{QIC}\left(\Pi,\ket{\Psi_{\textrm{c-c}}}_{A_{\textrm{IN}}B_{\textrm{IN}}R}\right)&=\frac{1}{2}\bigl[\textit{QLA}\left(\Pi,\ket{\Psi_{\textrm{c-c}}}_{A_{\textrm{IN}}B_{\textrm{IN}}R}\right)+\textit{QLB}\left(\Pi,\ket{\Psi_{\textrm{c-c}}}_{A_{\textrm{IN}}B_{\textrm{IN}}R}\right)\bigr]\\
&\geq n-\log n.
\end{align*}
\end{proof}

Finally, as a corollary of this result, by using the fact that the information cost of a quantum protocol on any input state lower bounds the number of qubits exchanged, we obtain Theorem~\ref{th:rand ent}.

\subsection{Proof of lower bound in Lemma~\ref{le:IP in phase}} \label{ap:phase lb}

In this Appendix, we adapt the techniques used in proving Theorem~\ref{th:rand ent} to prove the lower bound in Lemma~\ref{le:IP in phase}. As we are considering clean computation and we initialized the answer register in Eqs.~\eqref{eq:class-class} and \eqref{eq:class-sup} in the $\ket{-}$ state (so the result in the computation was stored in the phase), Lemmas \ref{le:inf flow} and \ref{le:cs vs ss} and Corollary \ref{co:clean} also apply to protocols for cleanly computing $\textit{IP}_{n}^{\textit{phase}}$ where there is no answer register.

We now modify the proof of Theorem~\ref{th:rand clean inf} to show that:
\begin{lemma} \label{le:IP leak}
Any protocol $\Pi$ that exactly computes $\textit{IP}_{n}^{\textit{phase}}$ cleanly (with or without pre-shared entanglement) is such that:
\begin{align}
\begin{split}
\textit{QLA}\left(\Pi,\ket{\Psi_{\textrm{c-c}}}_{A_{\textrm{IN}}B_{\textrm{IN}}R}\right)\geq n,\\
\textit{QLB}\left(\Pi,\ket{\Psi_{\textrm{c-c}}}_{A_{\textrm{IN}}B_{\textrm{IN}}R}\right)\geq n,
\end{split}
\end{align}
where here:
\begin{equation}
\ket{\Psi_{\textrm{c-c}}}_{A_{\textrm{IN}}B_{\textrm{IN}}R}=\frac{1}{2^n}\sum_{x,y}\ket{x}_{A}\ket{\vec{0}}_{A_0}\ket{y}_{B}\ket{\vec{0}}_{B_0}\ket{\Phi}_{A_EB_E}\ket{xy}_R.
\end{equation}
\end{lemma}
\begin{proof}
Consider running $\Pi$ on registers $A$, $B$, $A_0$, $B_0$ and $A_EB_E$ of the input state:
\begin{equation*}
\ket{\Psi_{\textrm{c-s}}}_{A_{\textrm{IN}}B_{\textrm{IN}}R}=\frac{1}{2^n}\sum_{x,y}\ket{x}_{A}\ket{\vec{0}}_{A_0}\ket{y}_{B}\ket{y}_{B'}\ket{\vec{0}}_{B_0}\ket{\Phi}_{A_EB_E}\ket{x}_R.
\end{equation*}
At the end of the protocol, the state will be:
\begin{align*}
\ket{\Psi_{\textrm{c-s}}}_{A_{\textrm{OUT}}B_{\textrm{OUT}}R}&=\sum_x\frac{1}{\sqrt{2^n}}\ket{x}_A\ket{\vec{0}}_{A_0}\left(\sum_y\frac{1}{\sqrt{2^n}}\left(-1\right)^{x\cdot y}\ket{y}_B\ket{y}_{B'}\right)\ket{\vec{0}}_{B_0}\ket{\Phi}_{A_EB_E}\ket{x}_R\\
&\equiv\sum_{x}\frac{1}{\sqrt{2^n}}\ket{x}_A\ket{\vec{0}}_{A_0}\ket{\psi_x}_{BB'}\ket{\vec{0}}_{B_0}\ket{\Phi}_{A_EB_E}\ket{x}_R,
\end{align*}
where:
\begin{equation*}
\ket{\psi_x}_{BB'}=\sum_y\frac{1}{\sqrt{2^n}}\left(-1\right)^{x\cdot y}\ket{y}_B\ket{y}_{B'}.
\end{equation*}

The reduced state on systems $B_{\textrm{OUT}}$ and $R$ after the protocol is:
\begin{equation*}
\rho_{B_{\textrm{OUT}}R}=\sum_x \frac{1}{2^n}\ketbra{\psi_x}{\psi_x}_{BB'}\otimes\ketbra{0}{0}_{B_0}\otimes\tr_{A_E}\left[\ketbra{\Phi}{\Phi}_{A_EB_E}\right]\otimes\ketbra{x}{x}_R.
\end{equation*}
Note that $I\left(B_{\textrm{IN}}:R\right)=0$ while, as per the proof of Theorem~\ref{th:rand clean inf}:
\begin{align*}
I\left(B_{\textrm{OUT}}:R\right)&=S\left(B_{\textrm{OUT}}\right)+S\left(R\right)-S\left(B_{\textrm{OUT}} R\right)\\
&=S\left(B_{\textrm{OUT}}\right)+S\left(R\right)-S\left(A_{\textrm{OUT}}\right)\\
&=S\left(BB'\right).
\end{align*}

Now, using Lemma~\ref{le:inf flow} and the fact that the quantum conditional mutual information is non-negative, we have:
\begin{equation*}
\textit{QLA}\left(\Pi,\ket{\Psi_{\textrm{c-s}}}_{A_{\textrm{IN}}B_{\textrm{IN}}R}\right)\geq S\left(\rho_{BB'}\right),
\end{equation*}
where $\rho_{BB'}=\frac{1}{2^n}\sum_{x}\ketbra{\psi_x}{\psi_x}_{BB'}$. Now, as $\left\{\ket{\psi_x}\right\}$ is an orthonormal set of vectors, we have:
\begin{equation*}
\textit{QLA}\left(\Pi,\ket{\Psi_{\textrm{c-s}}}_{A_{\textrm{IN}}B_{\textrm{IN}}R}\right)\geq n.
\end{equation*}

Using Lemma~\ref{le:cs vs ss} we now use this last expression to lower bound $\textit{QLA}\left(\Pi,\ket{\Psi_{\textrm{c-c}}}_{A_{\textrm{IN}}B_{\textrm{IN}}R}\right)$ by $n$. Corollary~\ref{co:clean} then gives that $\textit{QLB}\left(\Pi,\ket{\Psi_{\textrm{c-c}}}_{A_{\textrm{IN}}B_{\textrm{IN}}R}\right)\geq n$. 
\end{proof}

To prove the lower bound on $\textit{Q}_{\textit{clean}}\left(\textit{IP}_{n}^{\textit{phase}}\right)$, we now show that without any pre-shared entanglement $n$ qubits of communication are not sufficient to leak both $n$ bits from Alice to Bob and $n$ bits from Bob to Alice when the players are given independent uniformly distributed inputs.
\begin{lemma}
We have:
\begin{equation}
\textit{Q}_{\textit{clean}}\left(\textit{IP}_{n}^{\textit{phase}}\right)\geq n+1.
\end{equation}
\end{lemma}
\begin{proof}
As $0\leq\frac{1}{2}I\left(C:R|B\right)\leq\log\textrm{dim}\left(C\right)$, each qubit of communication can leak at most 2 bits of information between the players. Hence if we have a clean $n$-qubit protocol for $\textit{IP}_n$, by Lemma~\ref{le:IP leak} each qubit must leak at least 2 bits. We will argue that if the players do not initially share any entanglement then this cannot happen.

Let $\Pi$ be said $n$-qubit clean protocol. Consider the protocol $\Pi'$ constructed as follows:
\begin{enumerate}
\item Both players create a copy of their classical input.
\item The players run $\Pi$ on their original input, keeping the copies to one side.
\item The players cleanly erase their input copies.
\end{enumerate}
Obviously $\Pi'$ is also a clean $n$-qubit protocol for computing inner product in the phase and each qubit must still leak at least 2 bits of information.

If the players' inputs are uniformly distributed, after Step 1 the total, purified state of the protocol can be taken to be:
\begin{equation*}
\frac{1}{{2^n}}\sum_{x,y}\ket{x}_A\ket{\vec{0}}_{A_0}\ket{0}_{C}\ket{x}_{A'}\ket{y}_{B}\ket{y}_{B'}\ket{xy}_R,
\end{equation*}
where $A$ and $B$ denotes the input registers, $R$ their purification, $A'$ and $B'$ the copies, $C$ the first qubit that Alice (w.l.o.g.) will send to Bob and $A_0$ Alice's additional ancilla qubits. To determine her message, Alice now applies a unitary to her share of the qubits (excluding the copy) resulting in:
\begin{equation*}
\frac{1}{{2^n}}\sum_{x,y}U\left(\ket{x}_A\ket{\vec{0}}_{A_0}\ket{0}_{C}\right)\ket{x}_{A'}\ket{y}_{B}\ket{y}_{B'}\ket{xy}_R.
\end{equation*}
Now consider:
\begin{equation*}
I\left(C:R|BB'\right)=I\left(C:R\right)=S\left(C\right)-S\left(C|R\right).
\end{equation*}
Obviously $S\left(C\right)\leq1$ while $\rho_{RC}$ is a separable state so $S\left(C|R\right)$ is non-negative. Hence the first qubit of communication leaks strictly less than 2 bits of information and $\Pi'$ cannot have been a clean protocol that computed $\textit{IP}_{n}^{\textit{phase}}$.
\end{proof}

\end{document}